\documentclass[reprint,twocolumn,aps,amsmath,amssymb,pra,superscriptaddress,floatfix,10pt,tightenlines]{revtex4-2}
\usepackage{amsmath, amssymb, amsthm}
\usepackage{mathtools}
\usepackage{color}
\usepackage{braket}
\usepackage{comment}
\usepackage{bbm}
\usepackage{bm}
\usepackage[breaklinks=true]{hyperref}
\usepackage{url}
\usepackage{color}

\DeclareMathOperator{\tr}{Tr}

\DeclareMathOperator{\id}{\mathrm{id}}

\DeclareMathOperator{\entropy}{H}
\DeclareMathOperator{\fidelity}{F}
\newcommand{\Lin}{\mathrm{L}}

\newcommand{\Channel}{\mathrm{C}}
\newcommand{\Unitary}{\mathrm{U}}
\newcommand{\Density}{\mathrm{D}}

\newcommand{\proj}[1]{|#1 \rangle\langle#1|}

\allowdisplaybreaks[4]

\theoremstyle{plain}
\newtheorem{theorem}{Theorem}

\newtheorem{proposition}[theorem]{Proposition}
\theoremstyle{definition}
\newtheorem{definition}[theorem]{Definition}
\newtheorem{example}{Example}

\theoremstyle{remark}
\newtheorem{conjecture}{Conjecture}
\newtheorem{remark}[conjecture]{Remark}

\begin{document}
\author{Kohdai Kuroiwa}
\email{kkuroiwa@uwaterloo.ca}
\affiliation{Institute for Quantum Computing, University of Waterloo, Ontario, Canada, N2L 3G1}
\affiliation{Department of Physics and Astronomy, University of Waterloo}
\affiliation{Department of Combinatorics and Optimization, University of Waterloo}
\affiliation{Perimeter Institute for Theoretical Physics, Ontario, Canada, N2L 2Y5}

\author{Debbie Leung}
\email{wcleung@uwaterloo.ca}
\affiliation{Institute for Quantum Computing, University of Waterloo, Ontario, Canada, N2L 3G1}
\affiliation{Department of Combinatorics and Optimization, University of Waterloo}
\affiliation{Perimeter Institute for Theoretical Physics, Ontario, Canada, N2L 2Y5}

\date{\today}

\title{Rate Reduction of Blind Quantum Data Compression with Local Approximations Based on Unstable Structure of Quantum States}

\begin{abstract} 
    In this paper, we propose a new protocol for a data compression task, \textit{blind quantum data compression}, with finite local approximations. 
    The rate of blind data compression is susceptible to approximations even when the approximations are diminutive. 
    This instability originates from the sensitivity of a structure of quantum states against approximations, which makes the analysis of blind compression in the presence of approximations intractable. 
    In this paper, we constructed a protocol that takes advantage of the instability to reduce the compression rate substantially. 
    Our protocol shows a significant reduction in rate for specific examples we examined. 
    Moreover, we apply our methods to diagonal states, and propose two types of approximation methods in this special case. 
    We perform numerical experiments and observe that one of these two approximation methods performs significantly better than the other. 
    Thus, our analysis makes a first step toward general investigation of blind quantum data compression with the allowance of approximations towards further investigation of approximation-rate trade-off of blind quantum data compression. 
\end{abstract}

\maketitle
\section{Introduction}~\label{sec:introduction}
\textit{Data compression} is a fundamental information processing task
in information theory \cite{Shannon1948}.  In this task, the data is drawn from a known distribution identically and independently an
asymptotically large number of times, and the goal is to transmit the
data as efficiently as possible.  
The concepts of data, distributions, and data compression have natural 
generalizations in the quantum setting. 
\textit{Quantum data compression} under various setups has been extensively investigated~\cite{Schumacher1995,Richard1994,Barnum1996,Lo1995,Horodecki1998,Barnum2001,Dur2001,Kramer2001,Horodecki2000,Hayashi2006,Koashi2001,Bennett2001,Bennett2005,Abeyesinghe2006,Jain2002,Bennett2014a,Anshu2022,Khanian2020}. 
(See Sec.~\ref{subsec:review_compression} for more detailed discussions.)
Among all these setups, we focus on \textit{blind quantum data compression}, in which the sender does not know the description of the quantum states to be transmitted.
This problem has been widely studied~\cite{Schumacher1995,Richard1994,Barnum1996,Lo1995,Horodecki1998,Barnum2001,Kramer2001, Koashi2001,Anshu2022,Khanian2020}, 
and in particular, in Ref.~\cite{Koashi2001}, the structure of a set of quantum states is used to derive an optimal 
compression rate under the constraint that the error vanishes asymptotically.

More recently, Refs.~\cite{Anshu2022,Khanian2020} indicated the possibility that the error analysis conducted in Ref.~\cite{Koashi2001} can be further tightened. 
Indeed, the structure of quantum states used in the analysis is highly sensitive to approximations. 
Reference~\cite{Anshu2022} introduced a blind compression protocol for 
special quantum states called \textit{classical states}. 
The protocol achieves a considerably lower rate by allowing finite approximations. 
Previously, such data compression tasks with finite approximations 
have been investigated through \textit{quantum rate distortion theories}.
Here, a function called \textit{distortion measure} is employed to characterize possible approximations between the original and resulting states. 
A large body of results has been obtained for several setups of data compression and channel coding via rate distortion theories~\cite{Barnum2000,Devetak2001,Devetak2002,Datta2013a,Datta2013b,Datta2013c,Khanian_Kuroiwa2022}. 
Nevertheless, the approximate compression of general mixed quantum states has been largely unexplored.

In this paper, as a first step of a more thorough error analysis of blind quantum data compression with finite approximations, 
we propose and study a method to reduce the compression rate that exploits the sensitivity of the structure of quantum states under local approximation.  
In particular, in previous work \cite{Koashi2001}, the structure of quantum states define \textit{redundant information} which need not be transmitted.
We investigate how much this structure can be sensitive to approximations, by explicitly exhibiting ensembles that have information that is not exactly redundant, but approximately so.
In our examples, this approximately redundant information occupies a large dimension. 
The previous optimal protocol \cite{Koashi2001} is defined using the exact redundant information, and therefore may have a very high compression rate.
We instead propose a protocol that identifies and omits this large approximately redundant information, resulting in 
a substantial reduction of the compression rate.

We apply our protocol to classical states to show the rate-approximation relation of our protocol. 
We can interpret the generation of approximately redundant parts as a coarse-graining of probability distributions corresponding to the given classical states. 
More specifically, we propose two methods to generate approximately redundant parts: the \textit{arithmetic mean method} and the \textit{geometric mean method}.  
Comparing these two methods numerically, we gain an understanding of the trade-off between the approximation and the compression rate. 
Furthermore, we also numerically show the relation between the rate and the system dimension when the allowed approximation depends on the dimension. 
These numerical simulations show that the arithmetic mean method performs essentially better than the geometric mean method. 

In summary, we propose a new protocol for blind quantum data compression with finite local approximations. 
Our examples show the evidence that this protocol can show a substantial reduction in the compression rate.  
Furthermore, the numerical analysis of classical ensembles implies how the compression rate depends on approximations. 
These results shed light on a rigorous form of error-rate trade-off of blind data compression, 
leading to a further understanding of efficient quantum data transmission.

The rest of this paper is organized as follows. 
We first review the basic setup and previous results of blind quantum data compression in Sec.~\ref{sec:preliminary}. 
Then, we show and discuss our results on blind quantum data compression with finite local approximations in Sec.~\ref{sec:results_blind_compression}.
We construct a novel protocol that performs well compared with the previous one that works with asymptotically vanishing errors. We also show numerical experiments for blind compression of two-state classical ensembles. 
Finally, we summarize and discuss our results in Sec.~\ref{sec:Blind_discussions}. 

\section{Preliminaries}~\label{sec:preliminary}
In this section, we briefly review basic concepts for blind quantum data compression. 
In Sec.~\ref{subsec:notation}, we summarize relevant mathematical tools and notations. 
In Sec.~\ref{subsec:structure_quantum_ensemble}, we describe a structure of quantum ensembles, which is called \textit{Koashi-Imoto decomposition} (KI decomposition). 
In Sec.~\ref{subsec:review_compression}, we review the basic setups and previous results on quantum data compression. 

\subsection{Mathematical Background and Notation}\label{subsec:notation}
Throughout this paper, we use capital alphabets $A,B,\ldots$ to represent quantum systems. 
For a quantum system $A$, we let $\mathcal{H}_A$ denote the corresponding complex Hilbert space. 
A vector $\ket{\psi} \in \mathcal{H}_A$ is called a \emph{pure state} if it has unit norm. 
For convenience, we also say that $\mathcal{H}_A$ is a quantum system. 
We only consider finite-dimensional quantum systems in this paper. 
Given a quantum system $\mathcal{H}$, we let $d_{\mathcal{H}}$ denote the dimension of this system. 

For a quantum system $\mathcal{H}$, 
let $\Lin(\mathcal{H})$ denote the set of linear operators on $\mathcal{H}$. 
In particular, we define $I_{\mathcal{H}} \in \Lin(\mathcal{H})$ as the identity operator on $\mathcal{H}$. 
Density operators are positive semidefinite operators with unit trace in $\Lin(\mathcal{H})$.  
Moreover, let $\Density(\mathcal{H})$ denote the set of density operators on $\mathcal{H}$. 
When a density operator is rank $1$, it can be written as $\rho = \proj{\psi}$ using some state $\ket{\psi} \in \mathcal{H}$.  Such a density operator will also be called a pure state.  
In general, a density operator has a spectral decomposition which is a convex combination of pure states, therefore, it can be interpreted as a probabilistic mixture of pure states.  
Therefore, a density operator also represents a quantum state. 
If $\rho \in \Density(\mathcal{H})$ is not pure, it is called a mixed state. 
For a set of states on a quantum system $\mathcal{H}$, if they are simultaneously diagonalizable in a certain basis $\{|e_i\rangle\}_{i=1}^{d_{\mathcal{H}}}$, we can interpret each state as a distribution over the label $i$ for the basis.  We call such a set of states ``classical''.  
In this case, we may also call each state ``classical''.  

Given two quantum systems $\mathcal{H}_A$ and $\mathcal{H}_B$ with $d_{\mathcal{H}_A} \leq d_{\mathcal{H}_B}$, we use 
$\Unitary(\mathcal{H}_A,\mathcal{H}_B)$ to denote the set of isometries from $\mathcal{H}_A$ to $\mathcal{H}_B$. 
That is, for $U \in \Unitary(\mathcal{H}_A,\mathcal{H}_B)$, 
it holds that $UU^\dagger = I_{\mathcal{H}_A}$ where $\dagger$ represents the Hermitian conjugate. 
A linear map $\mathcal{N}:\Lin(\mathcal{H}_A) \to \Lin(\mathcal{H}_B)$ 
is called a quantum channel or a quantum operation if it is completely positive and trace-preserving. 
Let $\Channel(\mathcal{H}_A,\mathcal{H}_B)$ be the set of quantum channels from $\Lin(\mathcal{H}_A)$ to $\Lin(\mathcal{H}_B)$. 
In particular, we let $\id_{\mathcal{H}}$ denote the identity channel on $\mathcal{H}$. 

We can introduce a norm on $\Lin(\mathcal{H})$; let $\|\cdot\|_1$ be the trace norm. 
The trace norm is monotonic under quantum operations; 
that is, it holds that 
\begin{equation}
    \|\mathcal{N}(X)\|_1 \leqq \|X\|
\end{equation}
for all linear operators $X$ and quantum channels $\mathcal{N}$ applicable to $X$. 

For a quantum state $\rho \in \Density(\mathcal{H})$, the quantum entropy of $\rho$ is defined as 
\begin{equation}
    \entropy(\rho) \coloneqq -\tr(\rho\log_2\rho). 
\end{equation}
It is known that the quantum entropy is upper-bounded by the size of the system; 
that is, 
\begin{equation}
    \entropy(\rho) \leq \log_2 d_{\mathcal{H}}. 
\end{equation}

\subsection{Structure of Quantum Ensembles}\label{subsec:structure_quantum_ensemble}
We first review the definition of quantum ensembles, before describing 
a structure of them. 
\begin{definition}[Quantum Ensembles]
Let $\mathcal{H}$ be a quantum system, and let $\Sigma$ be an alphabet. A quantum ensemble $\Phi$ is a set of pairs of a positive real number and a quantum state
\begin{equation}
    \Phi \coloneqq \{(p_a,\rho_a) \in \mathbb{R}\times\Density(\mathcal{H}): a\in\Sigma\},
\end{equation}
where $\{p_a:a\in\Sigma\}$ forms a probability distribution; that is,
\begin{equation}
    0\leq p_a \leq 1 
\end{equation}
for all $a\in\Sigma$ and
\begin{equation}
    \sum_{a\in\Sigma} p_a = 1.
\end{equation}
We also write
\begin{equation}
    \Phi = \{p_a,\rho_a\}_{a\in\Sigma}.
\end{equation}
In addition, the average state $\rho_{\Phi}$ of a quantum ensemble $\Phi = \{p_a,\rho_a\}_{a\in\Sigma}$ is defined as
\begin{equation}
    \rho_{\Phi} \coloneqq \sum_{a\in\Sigma} p_a\rho_a. 
\end{equation}
\end{definition}

Here, we define a special class of quantum ensembles.
These are ensembles of classical states.  
\begin{definition}[Classical Ensembles]~\label{def:classical_ensemble}
Let $\Phi = \{p_a,\rho_a\}_{a\in\Sigma}$ be a quantum ensemble. 
The ensemble $\Phi$ is said to be a \textit{classical ensemble} 
if $\rho_a$ are simultaneously diagonalizable for all $a\in\Sigma$; that is, there exists an orthonormal basis of $\mathcal{H}$ such that each $\rho_a$ is diagonal.  
\end{definition}

We can also define another special class of quantum ensembles.  
A quantum ensemble $\{p_a,\rho_a\}_{a\in\Sigma}$ is called a \textit{pure-state ensemble}, if  
$\rho_a$ is pure for all $a\in\Sigma$.  
Suppose $\rho_a = \ket{\psi_a}\bra{\psi_a}$.  We may use the shorthand 
$\{p_a,\ket{\psi_a}\}_{a\in\Sigma}$ to represent the ensemble.  

Suppose that we have a quantum ensemble $\Phi = \{p_a,\rho_a\}_{a\in\Sigma}$. 
Reference~\cite{Koashi2002} gave a structure of a given quantum ensemble, which is called the \textit{Koashi-Imoto (KI) decomposition} or the \textit{KI structure}. 
Intuitively, when we have a quantum ensemble, we can decompose each state in the ensemble into the following three parts: \textit{classical part}, \textit{non-redundant quantum part}, and \textit{redundant part}. 
The following theorem provides a formal statement.  

\begin{theorem}[KI Decomposition~\cite{Koashi2002}]~\label{thm:KIdecomp}
Let $\mathcal{H}$ be a quantum system.
Let $\Phi = \{p_a,\rho_a\}_{a\in\Sigma}$ be an ensemble on the system.
Then, there exists a decomposition of the quantum system
\begin{equation}
    \mathcal{H} \coloneqq \bigoplus_{l \in \Xi} \mathcal{H}^{(l)}_Q \otimes \mathcal{H}^{(l)}_R 
\end{equation}
and a corresponding isometry 
\begin{equation}
\Gamma_{\Phi} \in \Unitary\left(\mathcal{H},\bigoplus_{l \in \Xi} \mathcal{H}^{(l)}_Q \otimes \mathcal{H}^{(l)}_R\right)
\end{equation}
satisfying the following conditions.
\begin{enumerate}
    \item 
    For all $a\in\Sigma$,
    \begin{equation}\label{eq:KIdecomp}
        \Gamma_{\Phi} \rho_a \Gamma_{\Phi}^\dagger = \bigoplus_{l \in \Xi} q^{(a,l)}\rho^{(a,l)}_Q\otimes \rho^{(l)}_R. 
    \end{equation}
    Here, for each $a \in \Sigma$, 
    $\{q^{(a,l)}:l \in \Xi\}$ forms a probability distribution over labels $l\in\Xi$,
    $\rho^{(a,l)}_Q \in \Density(\mathcal{H}^{l}_Q)$ is a density operator on system $\mathcal{H}^{l}_Q$, depending on both $a\in\Sigma$ and $l\in\Xi$, and $\rho^{(l)}_R \in \Density(\mathcal{H}^{l}_R)$ is a density operator on system $\mathcal{H}^{l}_R$, which is independent of $a\in\Sigma$. 

    \item 
    For each $l\in\Xi$,
    if a projection operator $P: \mathcal{H}^{(l)}_Q \to \mathcal{H}^{(l)}_Q$ satisfies  
    \begin{equation}
        P q^{(a,l)}\rho^{(a,l)}_Q = q^{(a,l)}\rho^{(a,l)}_Q P 
    \end{equation}
    for all $a \in \Sigma$, then $P = I_{\mathcal{H}^{(l)}_Q}$ or $P = 0$.
    
    \item 
    For all $l,l' \in \Xi$ such that $l\neq l'$,
    there exists no isometry $V \in \Unitary(\mathcal{H}^{(l)}_Q,\mathcal{H}^{(l')}_Q)$ such that
    \begin{equation}
        V q^{(a,l)}\rho^{(a,l)}_Q = \alpha q^{(a,l')}\rho^{(a,l')}_Q V
    \end{equation}
    with some positive real number $\alpha$ for all $a\in\Sigma$.
\end{enumerate}
\end{theorem}

The first statement in Theorem~\ref{thm:KIdecomp} shows the form of the KI decomposition. The second and third statements ensure that the decomposition~\eqref{eq:KIdecomp} is maximal; that is, we cannot further refine the structure. Indeed, the second one states that we cannot further decompose each block of the KI decomposition; the third one means that we cannot relate a block of the decomposition to another block. 
In the decomposition shown in Eq.~\eqref{eq:KIdecomp}, we can observe that all the quantum states in the given ensemble $\Phi$ can be decomposed in a block-diagonal structure. Note that $\rho^{(l)}_R \in \Density(\mathcal{H}^{(l)}_R)$ does not depend on $a\in\Sigma$; thus it is called redundant because it does not contain information about the label $a \in \Sigma$. 
Hereafter, when we specify the KI decomposition of a given ensemble $\Phi = \{p_a,\rho_a\}_{a\in\Sigma}$, we may omit $\Gamma_{\Phi}$ and write 
\begin{equation}
    \rho_a = \bigoplus_{l \in \Xi} q^{(a,l)}\rho^{(a,l)}_Q\otimes \rho^{(l)}_R
\end{equation}
for brevity.

Once the KI decomposition of a given ensemble is obtained, we can define two quantum channels that respectively correspond to removing and attaching the ensemble's redundant parts.
The proof of this theorem is shown in Appendix~\ref{proof_KIoperation}.

\begin{theorem}~\label{thm:KIoperations}
Let $\Phi = \{p_a,\rho_a\}_{a\in\Sigma}$ be an ensemble on a quantum system $\mathcal{H}$.  
Consider the KI decomposition of $\Phi$: 
\begin{equation}
    \mathcal{H} \coloneqq \bigoplus_{l \in \Xi} \mathcal{H}^{(l)}_Q \otimes \mathcal{H}^{(l)}_R 
\end{equation}
such that 
\begin{equation}
    \rho_a = \bigoplus_{l \in \Xi} q^{(a,l)}\rho^{(a,l)}_Q\otimes \rho^{(l)}_R.
\end{equation}
Then, there exist quantum channels $\mathcal{K}_{\mathrm{off}}$ and $\mathcal{K}_{\mathrm{on}}$ such that
\begin{align}
    &\mathcal{K}_{\mathrm{off}}(\rho_a) = \bigoplus_{l \in \Xi} q^{(a,l)}\rho^{(a,l)}_Q, \\
    &\mathcal{K}_{\mathrm{on}}\left(\bigoplus_{l \in \Xi} q^{(a,l)}\rho^{(a,l)}_Q\right) = \rho_a
\end{align}
for all $a\in\Sigma$.
\end{theorem}
Using $\mathcal{K}_{\mathrm{off}}$ and $\mathcal{K}_{\mathrm{on}}$, we can reversibly remove the redundant parts from the ensemble. 
In this paper, we refer to $\mathcal{K}_{\mathrm{off}}$ and $\mathcal{K}_{\mathrm{on}}$ as the \textit{KI operations}. 

\subsection{Review of Quantum Data Compression Problems}~\label{subsec:review_compression}
We discuss the basic setup for quantum data compression, the 
variations studied to-date, and known results about them. 

Quantum data compression is a quantum information processing task between two parties, the sender and the receiver.  In this task, the sender aims to transmit quantum data to the receiver as efficiently as possible, given the assumption that the data consists of asymptotically large number of states created independently.  The best rate captures fundamental quantum properties of the transmitted states.  

\begin{figure}
\centering
\includegraphics[width = 0.8\linewidth]{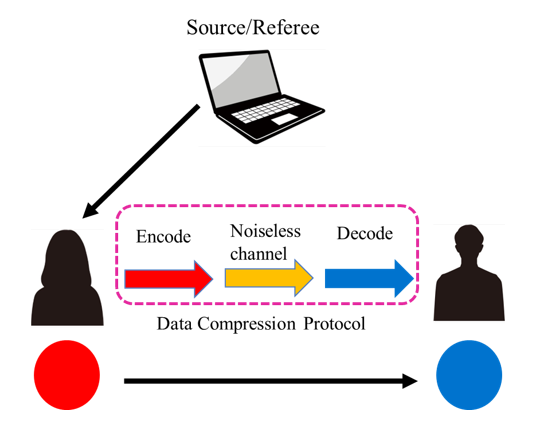}
\caption{Schematic diagram showing the setup of quantum data compression. In this task, the sender is given a target state (red-colored in the figure) from the referee. Then, the sender encodes the state and sends it to the receiver via a noiseless quantum channel. The receiver decodes the transmitted state to recover the original state (blue-colored in the figure).}
\label{fig:data_compression}
\end{figure}

We provide a more formal description below. 
Suppose that $\{p_x,\rho_x\}_{x\in\Sigma}$ is a quantum ensemble on system $A$.  
A referee draws $x \in \Sigma$ with probability $p_x$, writes down the outcome in a system $X$ (a classical system) with Hilbert space $\mathcal{H}_{X} \coloneqq \mathbb{C}^{\Sigma}$, and prepares $\rho_x$ in system $A$.  The resulting 
state is given by $\rho_{XA} \coloneqq \sum_{x\in\Sigma} p_x\ket{x}\bra{x}\otimes \rho_x$. 
Note that $\rho_A = \sum_{x\in\Sigma} p_x\rho_x$ is the average state of the ensemble. 
If the referee repeats the above process $n$ times \emph{independently},
then the resulting state is
\begin{align}
& \rho_{X_1 A_1} \otimes
\rho_{X_2 A_2} \otimes \cdots \otimes \rho_{X_n A_n}
\nonumber
\\[1ex]
= & \sum_{x_1,\cdots,x_n \in\Sigma} 
  p_{x_1} p_{x_2} \cdots p_{x_n} \ket{x_1}\bra{x_1} \otimes \cdots \otimes \ket{x_n}\bra{x_n} \hspace*{5ex}
\nonumber
\\[-2ex]
  & 
  \hspace*{23ex} 
  \otimes \rho_{x_1}\otimes \rho_{x_2}\otimes \cdots \otimes \rho_{x_n}
\end{align} 
where $|x_i\rangle\langle x_i|$ is a state on $X_i$ and $\rho_{x_i}$ 
is a state on $A_i$.  We use the shorthand 
$x^n$ for $x_1 x_2 \cdots x_n$, $p_{x^n}$ for $p_{x_1} p_{x_2} \cdots p_{x_n}$
and $\rho_{x^n}$ for 
$\rho_{x_1}\otimes \rho_{x_2}\otimes \cdots \otimes \rho_{x_n}$, and 
$\rho_{X^nA^n}$ for 
$\rho_{X_1 A_1} \otimes
\rho_{X_2 A_2} \otimes \cdots \otimes \rho_{X_n A_n}$.  
Then, the referee gives the systems $A_1 A_2 \cdots A_n$ to 
the sender.  
After receiving $A_1 \cdots A_n$ from the referee, the sender applies to
them an encoding channel $\mathcal{E}_n \in \Channel(\otimes_{i=1}^n
\mathcal{H}_{A_i},\mathcal{H}_{M_n})$.  
When the assistance of entanglement is allowed, the sender jointly 
encodes $A_1 \cdots A_n$ with her share of a given entangled state. 
Afterward, the sender transmits the 
system $M_n$ to the receiver through a noiseless quantum channel. 
The receiver decompresses the state by a decoding channel
$\mathcal{D}_n \in \Channel(\mathcal{H}_{M_n},\otimes_{i=1}^n
\mathcal{H}_{A'_i})$, 
where $A'_i\simeq A_i$ for all $i$.  
The final state 
$\id_{\mathcal{H}_X}^{\otimes n}\otimes(\mathcal{D}_n\circ\mathcal{E}_n)(\rho_{X^nA^n})$ 
has to resemble the initial state $\rho_{X^nA^n}$ in a certain way.  
The compression protocol is defined as a pair $(\mathcal{E}_n,\mathcal{D}_n)$ of the encoding and decoding
channels, but we also refer to their combined effect 
$\mathcal{D}_n\circ\mathcal{E}_n$ as the \textit{protocol}.

The compression rate $R_n$ of the protocol is given by the number of qubits transmitted 
divided by $n$: 
\begin{equation}
    R_n \coloneqq \frac{\log_2|M_n|}{n}. 
\end{equation}

In the basic setup shown above, 
the protocol and compression rate depend crucially on the following: 
\begin{enumerate}
    \item whether the given ensemble consists of pure states or mixed states;
    \item whether the referee tells the sender the label of the state to be compressed;
    \item what error criterion we adopt between the original state and the one recovered
      by the receiver. 
\end{enumerate}

In the second point mentioned above, if the sender knows the label, the
task is called \textit{visible compression}; 
otherwise, the task is called \textit{blind compression}.  
By definition, a blind compression protocol also works in the visible setting.  
Correspondingly, the compression rate in the blind setting is no less than that in the visible setting.  
On the third point, two definitions of error criteria have been considered: 
\textit{global error criterion} and \textit{local error criterion}. 
Under the global error criterion, a protocol satisfying 
\begin{equation}
    \left\|\id_{\mathcal{H}_X}^{\otimes n}\otimes(\mathcal{D}_n\circ\mathcal{E}_n)(\rho_{X^nA^n}) - \rho_{X^nA'^n}\right\|_1 \leq \epsilon_n; 
\end{equation}
is said to have global error $\epsilon_n$. 
In this case, the entire state emerges close to the entire original state.  
On the other hand, under the local error criterion, a protocol satisfying 
\begin{equation}
    \left\|\tr_{\bar{X}_k\bar{B}_k}(\id_{\mathcal{H}_X}^{\otimes n}\otimes(\mathcal{D}_n\circ\mathcal{E}_n)(\rho_{X^nA^n})) - \rho_{XA'}\right\|_1 \leq \epsilon_n 
\end{equation}
for all integers $1\leq k \leq n$ is said to have local error $\epsilon_n$.  
In this case, the initial and final states should be close letter-wisely.
In these definitions, $\tr_{\bar{(\cdot)}_k}$ denotes the partial trace over all systems other than the $k$th one. 
The global error criterion implies the local error criterion, but the converse does not always hold.

We consider an asymptotic scenario where the number $n$ of states is sufficiently large. 
We call a pair $(\mathcal{E}_n,\mathcal{D}_n)$ with an encoding channel $\mathcal{E}_n$ and a decoding channel $\mathcal{D}_n$ an $(n,R_n,\epsilon_n)$ code if the protocol $\mathcal{D}_n\circ\mathcal{E}_n$ yields the rate $R_n$ within error $\epsilon_n$ under the error criterion we adopt. 
Then, we say that the rate $R$ is achievable if for any $\epsilon>0$ and $\delta>0$, 
there exists 
a positive integer $n_0$ such that for all integers $n\geq n_0$,
we can construct an $(n,R+\delta,\epsilon)$ code.
Similarly, we say that the rate $R$ is achievable within error $\epsilon$ if for any $\delta>0$, 
there exists a positive integer $n_0$ such that for all integers $n\geq n_0$,  
we can construct an $(n,R+\delta,\epsilon)$ code.  

Next, we will summarize the state-of-the-art understanding of quantum data
compression in these setups.
Quantum data compression was first studied in Refs.~\cite{Schumacher1995,Richard1994,Barnum1996}.
These ground-breaking work focused on \textit{blind} compression of \textit{pure-state ensembles} under \textit{global error criterion}. 
It was shown that the optimal compression rate is the
quantum entropy of the average state of the given pure-state ensemble;
that is, letting $\{p_a,\ket{\psi_a}\}_{a\in\Sigma}$ denote a given
ensemble, the optimal rate is given by the entropy of the average state, 
\begin{equation}
    \entropy\left(\sum_{a\in\Sigma}p_a\ket{\psi_a}\bra{\psi_a}\right). 
\end{equation}
A particularly important compression protocol,
named after its inventor as \textit{Schumacher compression}, 
was proposed in Ref.~\cite{Schumacher1995} and 
reviewed in various textbooks such as Refs.~\cite{NielsenChuang2010,Wilde2017,Watrous2018}. 

After the aforementioned initial studies of blind compression of pure states, visible compression was subsequently considered. 
For a given quantum ensemble $\Phi = \{p_a,\rho_a\}_{a\in\Sigma}$, which is not necessarily pure, the optimal rate is lower-bounded by the Holevo information~\cite{Horodecki1998,Barnum2001} defined as
\begin{equation}
    I(\Phi) \coloneqq \entropy\left(\sum_{a\in\Sigma} p_a\rho_a\right) - \sum_{a\in\Sigma} p_a\entropy(\rho_a). 
\end{equation}
Observe that for a pure-state ensemble, the Holevo information is equal to the entropy of the average state since the entropy of a pure state is zero. 
Therefore, the optimal rate of the visible compression of a pure-state ensemble is also the entropy of the average state. 
Thus, surprisingly, for compression of a pure-state ensemble under the global error criterion, whether the compression is visible or blind does not make any difference to the optimal rate. 

After the extensive investigation of compression of pure-state ensembles, 
compression of general mixed-state ensemble became of central interest. 
Historically, the optimal rate of visible compression is called the \textit{effective information}; that of blind compression is called the \textit{passive information}. 
Letting $I_{e}(\Phi)$ denote the effective information of a quantum ensemble $\Phi$ and $I_{p}(\Phi)$ denote the passive information of a quantum ensemble $\Phi$, we have the following relation: 
\begin{equation}
    I(\Phi) \leq I_e(\Phi) \leq I_p(\Phi)
\end{equation}
since visible compression can be regarded as a subclass of blind compression. 
The quantity $I_d \coloneqq I_p -I_e$ is called the \textit{information defect}, and it characterizes a difference between the visible and blind compression tasks. 
In Ref.~\cite{Horodecki1998,Barnum2001}, a lower bound of the information defect was given, and Refs.~\cite{Dur2001,Kramer2001} showed examples for which the information defect is strictly positive. 
Moreover, the optimal rate of visible compression was derived in Ref.~\cite{Horodecki2000}, and the optimal rate is given by the entropy of an extension of a given state. 
In Ref.~\cite{Hayashi2006}, the authors gave another representation of the optimal rate of the visible compression. 
On the other hand, in Ref.~\cite{Koashi2001}, the optimal rate of blind compression was studied. The main difficulty of blind compression is that the sender does not know the label of a given state. The authors proposed a protocol in which the sender only sends the essential parts of a given ensemble, and the compression rate achieved by the protocol is given by the KI decomposition~\cite{Koashi2002}. 
This optimal protocol is based on a block coding of mixed states proposed in Ref.~\cite{Lo1995}. 
In fact, this compression rate is optimal under both the global and local error criteria. 

Here, let us note that we mainly discuss quantum data compression without any assistance; however, data compression tasks with assistance, e.g., entanglement and shared randomness, have also been widely investigated~\cite{Bennett2001,Bennett2005,Abeyesinghe2006,Jain2002,Bennett2014a,Anshu2022,Khanian2020} 

\section{Compression Protocol Using Finite Local Approximations}~\label{sec:results_blind_compression}
In this section, we show our results on the quantum blind compression with finite local approximations. 
In particular, we show our novel protocol and present some examples for which the protocol leads to large reduction of the compression rate compared to the previous results with asymptotically vanishing errors. 
Moreover, as a first step of general understandings of blind compression with finite approximations, we focus on quantum ensembles consisting of two classical states. 
Through numerical experiments, we reveal the performance of our protocol for classical ensembles. 

\subsection{Asymptotic Optimal Rate of Blind Quantum Data Compression}~\label{subsec:setup_blind_compression}
Blind data compression is a quantum information processing task between two parties, the sender and the receiver, in which fundamental quantum properties emerge. 
As also mentioned in Sec.~\ref{subsec:review_compression}, 
in blind data compression, 
the sender aims to asymptotically send quantum data \textit{without knowing its actual description} to the receiver as efficiently as possible. 

\begin{figure}
\centering
\includegraphics[width = \linewidth]{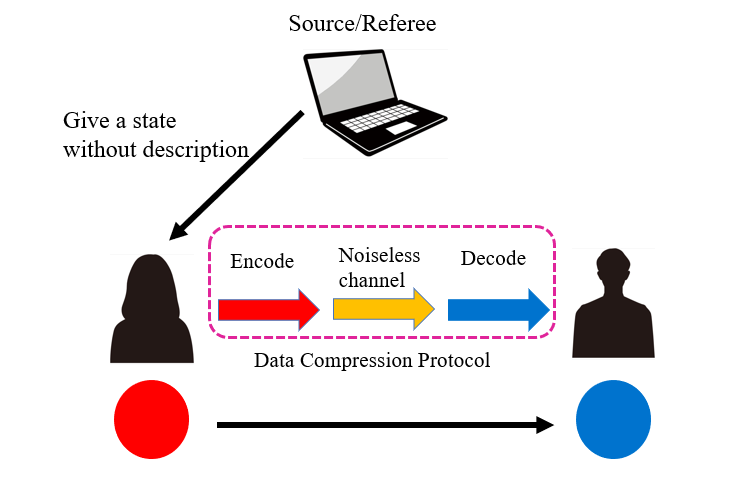}
\caption{Schematic diagram showing the setup of blind quantum data compression. In addition to the basic data compression setup, the sender cannot receive any information about the description or the label of a given state.}
\label{fig:blind_data_compression}
\end{figure}
Here we describe a formal procedure of blind quantum data compression 
following the setup shown in Sec.~\ref{subsec:review_compression}. 
Let $\{p_x,\rho_x\}_{x\in\Sigma}$ be a quantum ensemble the sender wants to send. 
The referee draws a state from the ensemble $n$ times in an independent and identically distributed (i.i.d.) manner, 
and let the resulting $n$-fold state be denoted by $\rho_{x^n}$.
Then, the referee gives the state $\rho_{x^n}$ to the sender. 
What is crucial in this setup is that the referee \textit{does not tell} the sender the label $x^n$ of the state. 
After receiving the state $\rho_{x^n}$ from the referee, the sender encodes the state by an encoding channel $\mathcal{E}_n$ and transmits the compressed state to the receiver through a noiseless quantum channel. 
The receiver decompresses the state by a decoding channel $\mathcal{D}_n$ to recover the initial state $\rho_{x_n}$. 
Note that they do not know the description of the given state while the sender has knowledge about the ensemble. 
Therefore, the compression scheme $(\mathcal{E}_n,\mathcal{D}_n)$ cannot be state-specific and must only depend on the ensemble. 

Now, we provide an overview of the achievable rates of blind data compression. 
In Ref.~\cite{Lo1995}, an achievable rate of blind state compression with a mixed state ensemble was investigated. 
It was figured out that the quantum entropy of the average state is an achievable rate even for a mixed state ensemble; 
that is, for a given ensemble $\{p_x,\rho_x\}_{x\in\Sigma}$, 
\begin{equation}\label{blind_compression_rate}
    \entropy\left(\sum_{x\in\Sigma}p_x\rho_x\right)
\end{equation}
is achievable. 
It had been open whether this rate is optimal or not~\cite{Lo1995,Horodecki1998,Barnum2001,Dur2001,Kramer2001} before Ref.~\cite{Koashi2001} conducted a more detailed analysis. 

In Ref.~\cite{Koashi2001}, the authors applied the Koashi-Imoto(KI) decomposition (Theorem~\ref{thm:KIdecomp} in Sec.~\ref{subsec:structure_quantum_ensemble})
to derive the optimal rate of blind compression. 
Recall that when an ensemble $\{p_x,\rho_x\}_x$ is given, 
we have the KI decomposition 
\begin{equation}
    \rho_x = \bigoplus_{l\in\Xi} q^{(x,l)} \rho_Q^{(x,l)} \otimes \rho_R^{(l)}. 
\end{equation}
Then, as shown in Theorem~\ref{thm:KIoperations},  
we can define quantum channels $\mathcal{K}_{\mathrm{off}}$ and $\mathcal{K}_{\mathrm{on}}$ satisfying
\begin{align}
    &\mathcal{K}_{\mathrm{off}} (\rho_x) 
    = \bigoplus_{l\in\Xi} q^{(x,l)} \rho_Q^{(x,l)},\\
    &\mathcal{K}_{\mathrm{on}} \left(\bigoplus_{l\in\Xi} q^{(x,l)} \rho_Q^{(x,l)}\right) 
    = \rho_x
\end{align}
for all $x\in\Sigma$. 
With these operations, the sender does not necessarily send the redundant parts of a target ensemble if the sender and the receiver agree on the ensemble in the scenario of blind data compression. 
If they both know the description of the ensemble, they also agree on $\mathcal{K}_{\mathrm{off}}$ and $\mathcal{K}_{\mathrm{on}}$; 
that is, they can freely take off and put on the redundant parts. 
Therefore, the sender only needs to send the classical and non-redundant quantum parts. 
By applying the argument of Ref.~\cite{Lo1995} after removing the redundant parts, 
we obtain an achievable rate 
\begin{equation}~\label{opt_blind_compression_rate}
    R^* = \entropy\left(\sum_{x\in\Sigma}p_x\bigoplus_{l\in\Xi} q^{(x,l)} \rho_Q^{(x,l)} \right).
\end{equation}
In addition, Ref.~\cite{Koashi2001} also proved that Eq.~\eqref{blind_compression_rate} is optimal even under local error criterion. 
To summarize, we have the following theorem. 
\begin{theorem}[{\cite{Koashi2001}}]
Let $\Phi = \{p_x,\rho_x\}_{x\in\Sigma}$ be a quantum ensemble with KI decomposition 
\begin{equation}
    \rho_x = \bigoplus_{l\in\Xi} q^{(x,l)} \rho_Q^{(x,l)} \otimes \rho_R^{(l)}. 
\end{equation}
Then, 
\begin{equation}
    R^* = \entropy\left(\sum_{x\in\Sigma}p_x\bigoplus_{l\in\Xi} q^{(x,l)} \rho_Q^{(x,l)} \right)
\end{equation}
is the optimal rate of blind compression of $\Phi$ for asymptotically vanishing errors under both the local and the global error criterion.  
\end{theorem}

To show the optimality of the rate, the authors of Ref.~\cite{Koashi2001} introduced two error functions of compression protocol $\Lambda_n \coloneqq (\mathcal{E}_n,\mathcal{D}_n)$:
\begin{align}
    \label{eq:error_f}
    f(\Lambda_n) &\coloneqq 1 - \sum_{x^n \in \Sigma^n}p_{x^n}\fidelity(\rho_{x^n},\Lambda_n(\rho_{x^n})) \\
    \label{eq:error_g}
    g(\Lambda_n) &\coloneqq h_2(\Delta)+\Delta\log_2(d_A-1),
\end{align}
with
\begin{equation}
    \Delta \coloneqq 1 - \sum_{i=1}^{d_A} \lambda_i(\rho)\braket{i|\Lambda_n(\ket{i}\bra{i})|i}. 
\end{equation}
Here, $\fidelity(\cdot,\cdot)$ is the fidelity function, and $\rho$ is the average state of the given ensemble with the spectral decomposition 
\begin{equation}
    \rho = \sum_{i=1}^{d_A}\lambda_i(\rho)\ket{i}\bra{i}. 
\end{equation}
Moreover, $h_2(\cdot)$ is the binary entropy function defined as 
\begin{equation}
    h_2(x) \coloneqq -x\log_2x - (1-x)\log_2(1-x)
\end{equation}
for $x \in [0,1]$. 
It was shown that when the protocol $\Lambda_n$ is given, the compression rate $R$ is
\begin{equation}
    R \geq R^* - g(\Lambda_n). 
\end{equation}
The authors proved the optimality by showing that when $f(\Lambda_n)$ goes to zero, $g(\Lambda_n)$ also converges to zero; that is, for an infinitesimal error, $R^*$ is the smallest compression rate that can be achievable.  

Recently, Ref.~\cite{Khanian2020} further investigated the optimality of the compression rate under the global error criterion. 
The authors introduced an error function that represents how an error allowed in the protocol affects the compression rate, 
and they revealed properties of this error function to study how the compression rate depends on an error. 
On the other hand, in this paper, we focus on local error criterion to investigate the sensitivity of the rate against approximations. 

\subsection{Our Protocol}~\label{subsec:protocol}
In this subsection, we motivate and describe our protocol for blind quantum data compression. 
Here, we consider local error criterion; that is, we only require the resulting state to be close to the original state letter-wisely. 
In addition, we allow finite approximations to the resulting state under this error criterion. 

We first summarize  
the compression scheme from the previous section: 
\begin{enumerate}
    \item After receiving a state $\rho_{x^n}$ from the referee, the sender applies $\mathcal{K}_{\mathrm{off}}$ letter-wisely to remove the redundant parts of the state and obtain $\mathcal{K}_{\mathrm{off}}^{\otimes n}(\rho_{x^n})$. 
    \item The sender encodes the state $\mathcal{K}_{\mathrm{off}}^{\otimes n}(\rho_{x^n})$ using an encoding channel $\mathcal{E}_n$ to obtain $\mathcal{E}_n\circ\mathcal{K}_{\mathrm{off}}^{\otimes n}(\rho_{x^n})$. \item The sender sends the state $\mathcal{E}_n\circ\mathcal{K}_{\mathrm{off}}^{\otimes n}(\rho_{x^n})$ via a noiseless quantum channel to the receiver. 
    \item The receiver decodes the transmitted state using a decoding channel $\mathcal{D}_n$ and obtains $\mathcal{D}_n\circ\mathcal{E}_n\circ\mathcal{K}_{\mathrm{off}}^{\otimes n}(\rho_{x^n})$. 
    \item The receiver put the redundant parts back using the quantum channel $\mathcal{K}_{\mathrm{on}}$; the resulting state is $\mathcal{K}_{\mathrm{on}}^{\otimes n}\circ\mathcal{D}_n\circ\mathcal{E}_n\circ\mathcal{K}_{\mathrm{off}}^{\otimes n}(\rho_{x^n})$. 
\end{enumerate}
Our main idea is as follows.  
We may use quantum channels $\Lambda_s^{(n)}$ and $\Lambda_r^{(n)}$, potentially acting jointly on the entire state, instead of $\mathcal{K}_{\mathrm{off}}^{\otimes n}$ and $\mathcal{K}_{\mathrm{on}}^{\otimes n}$.  
The encoding and decoding maps $\mathcal{E}_n$ and $\mathcal{D}_n$ may change accordingly.  
When we allow a finite, local, approximation of the given state $\rho_{x^n}$ by the resulting state $\Lambda_s^{(n)}\circ \mathcal{D}_n\circ\mathcal{E}_n\circ\Lambda_r^{(n)}(\rho_{x^n})$, a smaller rate may be achieved than the previous optimal rate with a vanishing error.
We show that even if $\Lambda_s^{(n)}$ and $\Lambda_r^{(n)}$ are letter-wise approximations, 
it is plausible that a reduction of the compression rate may be observed. 

\newcounter{dum}
\newenvironment{dumlist}{\begin{list}
{\arabic{dum}.}
{\usecounter{dum}}
}
{\end{list}}

In our protocol, we use the KI operations of a \emph{different} quantum ensemble 
on the given ensemble, and prove that this approximation procedure indeed performs well and reduces the rate.  
In more detail, suppose we are 
given a quantum ensemble $\{p_x,\rho_x\}_{x\in\Sigma}$ 
and the resulting state can differ letter-wisely from the original state up to $\epsilon>0$. 
We consider all possible approximate ensembles $\{p_x,\tilde{\rho}_x\}_{x\in\Sigma}$, 
such that $\tilde{\rho}_x$ is close to $\rho_x$ in the original given ensemble for each $x$. 
Each such approximate ensemble defines KI operations $\tilde{\mathcal{K}}_{\mathrm{off}}$ and $\tilde{\mathcal{K}}_{\mathrm{on}}$ as in Theorem~\ref{thm:KIoperations}.  
Now, suppose we use these KI operations $\tilde{\mathcal{K}}_{\mathrm{off}}^{\otimes n}$ and $\tilde{\mathcal{K}}_{\mathrm{on}}^{\otimes n}$ as 
an approximation scheme $\Lambda_s^{(n)}$ and $\Lambda_r^{(n)}$ for $\{p_x,\rho_x\}_{x\in\Sigma}$.  
How much error do we incur? If the error is small, what rate can be achieved? Before we answer
these questions, let us formally define our compression protocol with local approximations. 
\begin{dumlist}
\setcounter{dum}{-1}
    \item The sender and receiver pre-agree on an approximate ensemble $\{p_x,\tilde{\rho}_x\}_{x\in\Sigma}$ of the original given ensemble $\{p_x,\rho_x\}_{x\in\Sigma}$ such that for all $x\in\Sigma$,
    \begin{equation}~\label{eq:condition_approx_protocol}
        \left\| \; \tilde{\mathcal{K}}_{\mathrm{on}}\circ\tilde{\mathcal{K}}_{\mathrm{off}}(\rho_x) - \rho_x \; \right\|_1 \leq \epsilon \,,
    \end{equation}
    
    \item The referee gives a state $\rho_{x^n}$ generated by $n$ independent and identically distributed (i.i.d.) draws from $\{p_x,\rho_x\}_{x\in\Sigma}$ to the sender. 
    The sender applies $\tilde{\mathcal{K}}_{\mathrm{off}}^{\otimes n}$ to the given state, which yields $\tilde{\mathcal{K}}_{\mathrm{off}}^{\otimes n}(\rho_{x^n})$. 
    
    \item The sender applies an encoding operation $\mathcal{E}_n$ to obtain $\mathcal{E}_n\circ\tilde{\mathcal{K}}_{\mathrm{off}}^{\otimes n}(\rho_{x^n})$. 
    
    \item The sender transmits the encoded state $\mathcal{E}_n\circ\tilde{\mathcal{K}}_{\mathrm{off}}^{\otimes n}(\rho_{x^n})$ to the receiver via a noiseless quantum channel. 
    
    \item Upon receiving the state, the receiver applies the decoding channel $\mathcal{D}_n$ corresponding to $\mathcal{E}_n$, which yields $\mathcal{D}_n\circ\mathcal{E}_n\circ\tilde{\mathcal{K}}_{\mathrm{off}}^{\otimes n}(\rho_{x^n})$. 
    
    \item Finally, the receiver applies $\tilde{\mathcal{K}}_{\mathrm{on}}^{\otimes n}$ to obtain the resulting state $\tilde{\mathcal{K}}_{\mathrm{on}}^{\otimes n}\circ\mathcal{D}_n\circ\mathcal{E}_n\circ\tilde{\mathcal{K}}_{\mathrm{off}}^{\otimes n}(\rho_{x^n})$.  
\end{dumlist}

The resulting state $\tilde{\mathcal{K}}_{\mathrm{on}}^{\otimes n}\circ\mathcal{D}_n\circ\mathcal{E}_n\circ\tilde{\mathcal{K}}_{\mathrm{off}}^{\otimes n}(\rho_{x^n})$ of our protocol should be close to the original state $\rho_{x^n}$.
In step 2 of the protocol, we choose $\mathcal{E}_n$ to achieve blind compression of the ensemble $\{p_x,\tilde{\mathcal{K}}_{\mathrm{off}}(\rho_x)\}_{x\in\Sigma}$
with \emph{vanishing error}.  
That is, for any $\delta > 0$, there exists a positive integer $n_0$ such that for all integers $n\geq n_0$,  
\begin{equation}
\left\|\mathcal{D}_n\circ\mathcal{E}_n\circ\tilde{\mathcal{K}}_{\mathrm{off}}^{\otimes n}(\rho_{x^n}) - \tilde{\mathcal{K}}_{\mathrm{off}}^{\otimes n}(\rho_{x^n})\right\|_1 \leq \delta. 
\label{eq:compvanerr}
\end{equation}
Together with condition~\eqref{eq:condition_approx_protocol}, 
local error will be $\leq \epsilon + \delta$.  
In the rest of the paper, we omit $\delta$ considering that we may take $\delta$ arbitrarily small by choosing $n$ appropriately.

In our protocol, the condition~\eqref{eq:condition_approx_protocol} is crucial. 
Here, we give a sufficient condition for Eq.~\eqref{eq:condition_approx_protocol}; we prove that $\tilde{\mathcal{K}}_{\mathrm{off}}$ and $\tilde{\mathcal{K}}_{\mathrm{on}}$ satisfy the local error criterion for $\{p_x,\rho_x\}_{x\in\Sigma}$ if $\{p_x,\tilde{\rho}_x\}_{x\in\Sigma}$ is close enough to the original ensemble. 
Suppose that  $\{p_x,\rho_x\}_{x\in\Sigma}$ is a quantum ensemble, and let $\{p_x,\tilde{\rho}_x\}_{x\in\Sigma}$ be a quantum ensemble such that 
\begin{equation}
    \|\rho_x - \tilde{\rho}_x\|_1 \leq \frac{\epsilon}{2}
\end{equation}
for all $x$. 
Then, 
\begin{widetext}
\begin{align*}
\left\|\tilde{\mathcal{K}}_{\mathrm{on}}\circ\tilde{\mathcal{K}}_{\mathrm{off}}(\rho_x) - \rho_x\right\|_1
    &= \Bigg\|\left(\tilde{\mathcal{K}}_{\mathrm{on}}\circ\tilde{\mathcal{K}}_{\mathrm{off}}(\rho_x) - \tilde{\mathcal{K}}_{\mathrm{on}}\circ\tilde{\mathcal{K}}_{\mathrm{off}}
    (\tilde{\rho}_x)\right) +\left(\tilde{\mathcal{K}}_{\mathrm{on}}\circ\tilde{\mathcal{K}}_{\mathrm{off}}(\tilde{\rho}_x) - \rho_x\right) \Bigg\|_1 \\ 
    &\leq \left\|\tilde{\mathcal{K}}_{\mathrm{on}}\circ\tilde{\mathcal{K}}_{\mathrm{off}}(\rho_x) - \tilde{\mathcal{K}}_{\mathrm{on}}\circ\tilde{\mathcal{K}}_{\mathrm{off}}(\tilde{\rho}_x)\right\|_1 + \left\|\tilde{\mathcal{K}}_{\mathrm{on}}\circ\tilde{\mathcal{K}}_{\mathrm{off}}(\tilde{\rho}_x) - \rho_x\right \|_1 &\\
    &\leq 2\|\tilde{\rho}_x - \rho_x\|_1\\
    &\leq \epsilon. 
\end{align*}
\end{widetext}
The second line follows from the triangle inequality. 
In the third line, we used the monotonicity of the trace norm and the relation  $\tilde{\mathcal{K}}_{\mathrm{on}}\circ\tilde{\mathcal{K}}_{\mathrm{off}}(\tilde{\rho}_x) = \tilde{\rho}_x$. 
Thus, 
\begin{equation}
  \|\tilde{\mathcal{K}}_{\mathrm{on}}\circ\tilde{\mathcal{K}}_{\mathrm{off}}(\rho_x) - \rho_x\|_1 \leq \epsilon; 
\end{equation}
that is, a quantum ensemble for which each state is $\epsilon/2$-close to the state in the original ensemble works as an approximate ensemble for our protocol. 

Second, we analyze the rate performance of our protocol.  
The compression rate of our protocol is determined by $\mathcal{E}_n$; a  
rate of 
\begin{equation}
    R = \entropy\left(\sum_{x\in\Sigma}p_x\tilde{\mathcal{K}}_{\mathrm{off}}(\rho_x)\right)
\end{equation}
is sufficient to ensure vanishing error condition (\ref{eq:compvanerr}).  
In step 0 of the protocol, the sender and the receiver can minimize the above 
rate by optimizing the approximate ensemble: 
    \begin{equation}
        ~~~~~\entropy\left(\sum_xp_x\tilde{\mathcal{K}}_{\mathrm{off}}(\rho_x)\right) \leq \entropy\left(\sum_xp_x\mathcal{K}_{\mathrm{off}}(\rho_x)\right).
    \end{equation}

In this protocol, we need to find a good approximate ensemble to reduce the compression rate. 
While this optimization may not be intractable, it suffices to find good enough approximate ensemble to substantially reduce the rate, as we will see in the next section. 

\subsection{Reduction of rates with a finite error}
Here, we present two main examples for which our protocol performs better than the case where finite approximations are not allowed. 
The first example is over a small dimensional system, and it illustrates the procedure and performance of our protocol in an intuitive way; the second example is over a growing dimension, and exhibits a growing reduction of the compression rate.

First, we consider a four-dimensional two-state ensemble to intuitively understand the protocol. 
\begin{example}~\label{Example1_blindcomp}
Let $\epsilon>0$ be a fixed positive number sufficiently smaller than $1/2$. 
Consider the following two density operators. 
\begin{align}
    \rho_1 &:= 
    \frac{1}{4}\left(
    \begin{array}{cccc}
    2 - \tfrac{\epsilon}{2}  & 1 - \epsilon & 0 & 0\\
    1 - \epsilon & 2 - \tfrac{3\epsilon}{2} & 0 & 0\\
    0 & 0 & \tfrac{3\epsilon}{2} & 0\\
    0 & 0 & 0 & \tfrac{\epsilon}{2}
    \end{array}
    \right),
    \\
    \rho_2 &:= 
    \frac{1}{4}\left(
    \begin{array}{cccc}
    \epsilon & 0 & 0 & 0\\
    0 & \epsilon & 0 & 0\\
    0 & 0 & 2 - \epsilon & 1\\
    0 & 0 & 1 & 2 - \epsilon
    \end{array}
    \right).
\end{align}

We consider an ensemble $\{p_x,\rho_x\}_{x=1,2}$ where $p_1 = p_2 = \tfrac{1}{2}$. 
Note that this ensemble has no redundant parts while $\rho_1$ and $\rho_2$ are written in a block-diagonal form. 
Here, we construct an approximate ensemble $\{p_x,\tilde{\rho}_x\}_{x=1,2}$ as 
\begin{align}
    \tilde{\rho}_1 &:= 
    \frac{1}{4}\left(
    \begin{array}{cccc}
    2 & 1 & 0 & 0\\
    1 & 2 & 0 & 0\\
    0 & 0 & 0 & 0\\
    0 & 0 & 0 & 0
    \end{array}
    \right),\\
    \tilde{\rho}_2 &:= 
    \frac{1}{4}\left(
    \begin{array}{cccc}
    0 & 0 & 0 & 0\\
    0 & 0 & 0 & 0\\
    0 & 0 & 2 & 1\\
    0 & 0 & 1 & 2
    \end{array}
    \right).
\end{align}
Observe that $\tilde{\rho}_1$ and $\tilde{\rho}_2$ can be written as 
\begin{align}
    \tilde{\rho}_1 &\coloneqq  \ket{0}\bra{0}\otimes
        \frac{1}{4}\left(
        \begin{array}{cc}
        2 & 1\\
        1 & 2
        \end{array}
        \right),\\
    \tilde{\rho}_2 &\coloneqq  \ket{1}\bra{1}\otimes
        \frac{1}{4}\left(
        \begin{array}{cc}
        2 & 1\\
        1 & 2
        \end{array}
        \right)
\end{align}
Therefore, 
\begin{equation}\label{eq:redundant_ex1}
    \omega \coloneqq \frac{1}{4}\left(
        \begin{array}{cc}
        2 & 1\\
        1 & 2
        \end{array}
        \right)
\end{equation}
is a redundant part of the approximate ensemble. 
Then, we can define the KI operations corresponding to this redundant part as follows. 
\begin{align}
\tilde{\mathcal{K}}_{\mathrm{off}} &(\cdot) \coloneqq \tr_{2}(\cdot),\\
\tilde{\mathcal{K}}_{\mathrm{on}} &(\cdot) \coloneqq ((\ket{0}\bra{0}\cdot\ket{0}\bra{0}) + (\ket{1}\bra{1}\cdot\ket{1}\bra{1})) \otimes \omega \,, 
\end{align}
where the partial trace $\tr_{2}$ is taken over the second qubit system.  
With these operations, 
\begin{align*}
    \tilde{\mathcal{K}}_{\mathrm{off}}(\rho_1) 
    &= \left(
        \begin{array}{cc}
        1 - \tfrac{\epsilon}{2} & 0\\
        0 & \tfrac{\epsilon}{2}
        \end{array}
        \right)
    = \left(1 - \frac{\epsilon}{2}\right)\ket{0}\bra{0} + \frac{\epsilon}{2}\ket{1}\bra{1},\\
    \tilde{\mathcal{K}}_{\mathrm{off}}(\rho_2) 
    &= \left(
        \begin{array}{cc}
        \tfrac{\epsilon}{2} & 0\\
        0 & 1 - \tfrac{\epsilon}{2}
        \end{array}
        \right)
    = \frac{\epsilon}{2}\ket{0}\bra{0} + \left(1 - \frac{\epsilon}{2}\right)\ket{1}\bra{1}.
\end{align*}
In addition, 
\begin{align*}
\tilde{\mathcal{K}}_{\mathrm{on}}\circ\tilde{\mathcal{K}}_{\mathrm{off}}(\rho_1) 
    &= \frac{1}{8}\left(
        \begin{array}{cccc}
        4 -2\epsilon & 2-\epsilon & 0 & 0\\
        2-\epsilon & 4-2\epsilon & 0 & 0\\
        0 & 0 & 2\epsilon & \epsilon\\
        0 & 0 & \epsilon & 2\epsilon
        \end{array}
        \right) \\
    &= \left(\left(1 - \frac{\epsilon}{2}\right)\ket{0}\bra{0} + \frac{\epsilon}{2}\ket{1}\bra{1}\right)\otimes \omega, \\
    \tilde{\mathcal{K}}_{\mathrm{on}}\circ\tilde{\mathcal{K}}_{\mathrm{off}}(\rho_2) 
    &= \frac{1}{8}\left(
    \begin{array}{cccc}
    2\epsilon & \epsilon & 0 & 0\\
    \epsilon & 2\epsilon & 0 & 0\\
    0 & 0 & 4 - 2\epsilon & 2-\epsilon\\
    0 & 0 & 2 - \epsilon & 4 - 2\epsilon
    \end{array}
    \right)\\
    &= \left(\frac{\epsilon}{2}\ket{0}\bra{0} + \left(1 - \frac{\epsilon}{2}\right)\ket{1}\bra{1}\right)\otimes \omega.
\end{align*}
Hence, 
\begin{align*}
    \|\tilde{\mathcal{K}}_{\mathrm{on}}\circ\tilde{\mathcal{K}}_{\mathrm{off}}(\rho_1) - \rho_1\|_1 
    &= \left\|
    \frac{1}{8}\left(
    \begin{array}{cccc}
    -\epsilon & \epsilon & 0 & 0 \\
    \epsilon & \epsilon & 0 & 0 \\
    0 & 0 & -\epsilon & \epsilon\\
    0 & 0 & \epsilon & \epsilon 
    \end{array}
    \right)
    \right\|_1 \\ 
    &= \frac{\sqrt{2}\epsilon}{2} \leq \epsilon,\\
    \|\tilde{\mathcal{K}}_{\mathrm{on}}\circ\tilde{\mathcal{K}}_{\mathrm{off}}(\rho_2) - \rho_2\|_1 
    &= \left\|
    \frac{1}{8}\left(
    \begin{array}{cccc}
    0 & \epsilon & 0 & 0 \\
    \epsilon & 0 & 0 & 0 \\
    0 & 0 & 0 & -\epsilon\\
    0 & 0 & -\epsilon & 0 
    \end{array}
    \right)
    \right\|_1 \\
    &= \frac{\epsilon}{2} \leq \epsilon. 
\end{align*}
Therefore, the KI operations $\tilde{\mathcal{K}}_{\mathrm{off}}$ and $\tilde{\mathcal{K}}_{\mathrm{on}}$ satisfy 
the condition~\eqref{eq:condition_approx_protocol}. 

We now see that this example exhibits a reduction of the rate compared to the compression rate under asymptotically vanishing errors. 
Let $R_0$ denote the optimal rate for $\{p_x,\rho_x\}_{x\in\Sigma}$; let $R$ denote the rate for $\{p_x,\rho_x\}_{x\in\Sigma}$ obtained by $\tilde{\mathcal{K}}_{\mathrm{off}}$ and $\tilde{\mathcal{K}}_{\mathrm{on}}$. 
Then, we have 
\begin{align}
    R_0 &= \entropy\left(\frac{1}{2}\rho_1 + \frac{1}{2}\rho_2\right) \approx \log_24 = 2 ,\\
    R &= \entropy\left(\frac{1}{2}\tilde{\mathcal{K}}_{\mathrm{off}}(\rho_1) + \frac{1}{2}\tilde{\mathcal{K}}_{\mathrm{off}}(\rho_2)\right) \approx \log_22 = 1. 
\end{align}
Thus, this approximation reduces the compression rate to about half of the original rate. 
\end{example}

Next, we show an example for which a finite approximation dramatically changes the KI structure of a given ensemble.
For this example, we can see a large reduction of the compression rate compared with the compression rate under asymptotically vanishing errors. 
\begin{example}~\label{Example2_blindcomp}
Let $0 < \epsilon < 1/2$ be a fixed positive number. 
Let $\omega_a \in \mathcal{D}(\mathcal{H}_a)$ and $\omega_b \in \mathcal{D}(\mathcal{H}_b)$ be density operators. 
Let $N$ be a positive integer, and define $2N$-dimensional density operators
\begin{align*}
    \sigma_1 &\coloneqq 
    \frac{1}{4N}\left(
    \begin{array}{ccccc}
    2 & \epsilon \mathrm{e}^{i\alpha} & 0 & \cdots & \epsilon \mathrm{e}^{-i\alpha} \\
    \epsilon\mathrm{e}^{-i\alpha} & 2 & \epsilon\mathrm{e}^{i\alpha} & \cdots & 0 \\
    \vdots &\ddots & \ddots & \ddots & \vdots\\
    0 & \cdots & \epsilon\mathrm{e}^{-i\alpha}& 2 & \epsilon\mathrm{e}^{i\alpha} \\
    \epsilon\mathrm{e}^{i\alpha} & 0 & \cdots & \epsilon\mathrm{e}^{-i\alpha} & 2 
    \end{array}
    \right)
    \\
    \sigma_2 &\coloneqq 
    \frac{1}{4N}\left(
    \begin{array}{cccccc}
    1+2\epsilon &        &   &   &        &   \\
      & \ddots &   &   &        &   \\
      &        & 1 + 2\epsilon&   &        &   \\
      &        &   & 3 -2\epsilon&        &   \\
      &        &   &   & \ddots &   \\
      &        &   &   &        & 3-2\epsilon  
    \end{array}
    \right)
\end{align*}
with $0 < \alpha < 1/(4N)$. 
Here, the off-diagonal elements of $\sigma_2$ are all zero, and we omit writing these elements in the matrix form above. 
The states $\sigma_1$ and $\sigma_2$ do not have redundant parts. 
We show the proof in Appendix~\ref{appendix:example2}. 

With these density operators, let us define a quantum ensemble $\{p_x,\rho_x\}_{x=1,2}$ with 
\begin{align}
    &p_1 = p_2 = \frac{1}{2},\\
    \label{eq:KI_ex2_1}
    &\rho_1 \coloneqq \frac{1}{3}\omega_a \oplus \frac{1}{3}\sigma_1 \oplus \frac{1}{3}\omega_b, \\
    \label{eq:KI_ex2_2}
    &\rho_2 \coloneqq \frac{1}{6}\omega_a \oplus \frac{1}{3}\sigma_2 \oplus \frac{1}{2}\omega_b, 
\end{align}
where $\omega_a$ and $\omega_b$ are density operators so that the KI decomposition of $\rho_1$ and $\rho_2$ is given by Eqs.~\eqref{eq:KI_ex2_1} and \eqref{eq:KI_ex2_2}. 
Note that $\rho_1$ and $\rho_2$ already have redundant parts $\omega_a$ and $\omega_b$. 
Now, we introduce an approximate ensemble of $\{p_x,\tilde{\rho}_x\}_{x=1,2}$ with 
\begin{align}
    \tilde{\rho}_1 &\coloneqq \frac{1}{3}\omega_a \oplus \frac{1}{3}\tilde{\sigma_1} \oplus \frac{1}{3}\omega_b, \\
    \tilde{\rho}_2 &\coloneqq \frac{1}{6}\omega_a \oplus \frac{1}{3}\tilde{\sigma}_2 \oplus \frac{1}{2}\omega_b,
\end{align}
where 
\begin{align}
    \tilde{\sigma}_1 &\coloneqq 
    \frac{1}{2N}\left(
    \begin{array}{ccccc}
    1 & 0 & 0 & \cdots & 0\\
    0 & 1 & 0 & \cdots & 0\\
    \vdots &\ddots & \ddots & \ddots & \vdots\\
    0 & \cdots & 0& 1 & 0 \\
    0 & 0 & \cdots & 0 & 1
    \end{array}
    \right),
    \\
    \tilde{\sigma}_2 &\coloneqq 
    \frac{1}{4N}\left(
    \begin{array}{cccccc}
    1 &        &   &   &        &   \\
      & \ddots &   &   &        &   \\
      &        & 1&   &        &   \\
      &        &   & 3&        &   \\
      &        &   &   & \ddots &   \\
      &        &   &   &        & 3  
    \end{array}
    \right).
\end{align}
Then, $\tilde{\rho}_1$ and $\tilde{\rho}_2$ can be written as
\begin{align}
    \tilde{\rho}_1 &= \frac{1}{2} \tilde{\omega}_a \oplus \frac{1}{2} \tilde{\omega}_b, \\
    \tilde{\rho}_2 &= \frac{1}{4} \tilde{\omega}_a \oplus \frac{3}{4} \tilde{\omega}_b,  
\end{align}
where
\begin{align}
    \label{eq:redundant_ex2_1}
    \tilde{\omega}_a &\coloneqq \frac{2}{3}\left(\omega_a \oplus \frac{1}{2N}I\right),\\
    \label{eq:redundant_ex2_2}
    \tilde{\omega}_b &\coloneqq \frac{2}{3} \left(\frac{1}{2N}I\oplus\omega_b\right)
\end{align}
with the $N$-dimensional identity operator $I$. 
Note that the approximation vastly changes the structure of the KI decomposition. 
The original KI decomposition of $\{p_x,\rho_x\}_{x=1,2}$ consists of three classical parts. 
However, after the approximation, 
the KI decomposition of $\{p_x,\tilde{\rho}_x\}_{x=1,2}$ consists of two large blocks 
while the approximation changes the states only by $\epsilon$. 
Then, considering the KI operations $\tilde{\mathcal{K}}_{\mathrm{off}}$ and $\tilde{\mathcal{K}}_{\mathrm{on}}$ corresponding to this structure, we have
\begin{align}
    \tilde{\mathcal{K}}_{\mathrm{off}}(\rho_1)
    &= \frac{1}{2}
    \begin{pmatrix}
    1 & 0\\
    0 & 1\\
    \end{pmatrix},\\
    \tilde{\mathcal{K}}_{\mathrm{off}}(\rho_2)
    &= \frac{1}{12}
    \begin{pmatrix}
    3+2\epsilon & 0\\
    0 & 9-2\epsilon\\
    \end{pmatrix},
\end{align}
and 
\begin{align}
     \tilde{\mathcal{K}}_{\mathrm{on}}\circ\tilde{\mathcal{K}}_{\mathrm{off}}(\rho_1) &= \frac{1}{2} \tilde{\omega}_a \oplus \frac{1}{2} \tilde{\omega}_b = \tilde{\rho}_1 \\
    \tilde{\mathcal{K}}_{\mathrm{on}}\circ\tilde{\mathcal{K}}_{\mathrm{off}}(\rho_2) &= \left(\frac{1}{4}+\frac{\epsilon}{6}\right) \tilde{\omega}_a \oplus \left(\frac{3}{4}-\frac{\epsilon}{6}\right) \tilde{\omega}_b. 
\end{align}
Therefore, it holds that
\begin{align*}
    &\left\|\tilde{\mathcal{K}}_{\mathrm{on}}\circ\tilde{\mathcal{K}}_{\mathrm{off}}(\rho_1) - \rho_1\right\|_1 \leq \frac{\epsilon}{2N}\sum_{j = 0}^{2N-1} \left|\cos\left(\frac{2\pi j}{2N} + \alpha\right)\right| \leq \epsilon,\\
    &\left\|\tilde{\mathcal{K}}_{\mathrm{on}}\circ\tilde{\mathcal{K}}_{\mathrm{off}}(\rho_2) - \rho_2\right\|_1 = \frac{4\epsilon}{9} \leq \epsilon,
\end{align*}
implying that this pair of KI operations yields a compression protocol within a finite error $\epsilon$.   

Letting $R_0$ denote the optimal rate for $\{p_x,\rho_x\}_x$ without any approximations and $R$ denote the rate for $\{p_x,\rho_x\}_x$ obtained by $\tilde{\mathcal{K}}_{\mathrm{off}}$ and $\tilde{\mathcal{K}}_{\mathrm{on}}$, 
we have 
\begin{align}
    R_0 &= \entropy\left(\frac{1}{2}\mathcal{K}_{\mathrm{off}}(\rho_1) + \frac{1}{2}\mathcal{K}_{\mathrm{off}}(\rho_2)\right) \gtrsim \log_2 N ,\\
    R &= \entropy\left(\frac{1}{2}\tilde{\mathcal{K}}_{\mathrm{off}}(\rho_1) + \frac{1}{2}\tilde{\mathcal{K}}_{\mathrm{off}}(\rho_2)\right) \leq \log_22 = 1 ,
\end{align}
which shows the reduction of the compression rate to a constant rate, which is independent of the size of the system. 
\end{example} 

Thus, our compression protocol achieves a better compression rate 
by artificially generating approximate redundant parts. 
In the first example, we can generate a redundant part $\omega$ defined in Eq.~\eqref{eq:redundant_ex1}; 
in the second example, we can generate large redundant parts $\tilde{\omega}_a$ and $\tilde{\omega}_b$ shown in Eqs.~\eqref{eq:redundant_ex2_1} and \eqref{eq:redundant_ex2_2}. 
Our protocol has a much better rate when a quantum ensemble has an approximate ensemble with large redundant parts. 
By not sending these approximate redundant parts, we can achieve significantly small compression rates. 
Remarkably, in Example~\ref{Example2_blindcomp}, we can see that even a small approximation leads to a compression rate independent of the dimension of the system. 

Despite these examples showing large reductions in rates of blind compression, we would like to remark that we do not necessarily find a good approximation. 
For example, if an allowed error is much smaller than $\epsilon$ in Examples~\ref{Example1_blindcomp} and \ref{Example2_blindcomp}, we cannot apply the same approximation anymore. 
To further advance the study of blind quantum compression with finite approximations, a general investigation of conditions under which we can successfully find a good approximate ensemble of a given ensemble is needed. 
It would also be interesting to study a more general compression protocol for blind quantum compression with finite approximations. 

\begin{remark}
Let us illustrate the significance of the local error criterion in our examples. 
In this remark, we consider the fidelity 
$\fidelity(\rho,\sigma) \coloneqq \left\|\sqrt{\rho}\sqrt{\sigma}\right\|_1$ as a measure of accuracy instead of the trace norm to make the analysis easier. 
Note that these two measures are equivalent in the sense of the Fuchs-van de Graaf inequality~\cite{Watrous2018} 
\begin{equation}~\label{eq:van_de_graaf}
    1 - \frac{1}{2}\|\rho - \sigma\|_1 \leq \fidelity(\rho,\sigma) \leq \sqrt{1 - \frac{1}{4}\|\rho - \sigma\|_1^2}; 
\end{equation}
in particular, $\|\rho - \sigma\|_1 = 2$ if and only if $\fidelity(\rho,\sigma) = 0$, and $\|\rho - \sigma\|_1 = 0$ if and only if $\fidelity(\rho,\sigma) = 1$. 

In both of our examples, the resulting states $\tilde{\mathcal{K}}_{\mathrm{on}}\circ\tilde{\mathcal{K}}_{\mathrm{off}}(\rho_1)$ and $\tilde{\mathcal{K}}_{\mathrm{on}}\circ\tilde{\mathcal{K}}_{\mathrm{off}}(\rho_2)$ satisfy 
\begin{align}
    \fidelity(\rho_i, \tilde{\mathcal{K}}_{\mathrm{on}}\circ\tilde{\mathcal{K}}_{\mathrm{off}}(\rho_i)) \leq \sqrt{1 - O(\epsilon^2)}. 
\end{align}
Since the fidelity is multiplicative to product-state inputs~\cite{Watrous2018}, 
for any $i_1,\ldots,i_n \in \{1,2\}$, 
\begin{equation*}
\begin{aligned}
    &\fidelity(\rho_{i_1}\otimes \cdots \otimes \rho_{i_n}, \tilde{\mathcal{K}}_{\mathrm{on}}\circ\tilde{\mathcal{K}}_{\mathrm{off}}(\rho_{i_1}) \otimes \cdots \otimes \tilde{\mathcal{K}}_{\mathrm{on}}\circ\tilde{\mathcal{K}}_{\mathrm{off}}(\rho_{i_n})) \\ 
    &= \prod_{k=1}^n \fidelity(\rho_{i_k}, \tilde{\mathcal{K}}_{\mathrm{on}}\circ\tilde{\mathcal{K}}_{\mathrm{off}}(\rho_{i_k})) \\ 
    &\leq \prod_{k=1}^n \sqrt{1 - O(\epsilon^2)} 
    = \left(\sqrt{1 - O(\epsilon^2)}\right)^n 
    \xrightarrow{n\to\infty} 0. 
\end{aligned}
\end{equation*}
Thus, our protocol is asymptotically inaccurate in terms of the global error criterion though the local error is small even when $n\to\infty$. 
We believe that the rate reduction we achieved is realized by sacrificing the global accuracy. 
This intuition is closely related to the problem of \textit{strong converse}~\cite{Winter1999}. 
The strong converse of the blind data compression of mixed states is still largely unexplored, and we do not discuss this topic in this paper. 
\end{remark}

\subsection{Approximation of Classical Ensembles}~\label{subsec:approximation_classical_ensemble}
In this section, we discuss compression of classical ensembles under finite approximations, aiming to obtain insight for a more general setting. 
Here, to see properties of blind compression of classical ensembles concisely, we consider classical ensembles consisting of two states. 
We propose two approximation methods for a two-state classical ensemble, namely, the \textit{arithmetic mean method} and the \textit{geometric mean method}. 

As seen in Definition~\ref{def:classical_ensemble}, a classical ensemble is associated with a fixed basis in which all the states in the ensemble are diagonal.  
Indeed, when we approximate a classical ensemble with respect to the fixed basis, we only have to consider an approximation of diagonal elements. 
See Appendix~\ref{appendix:classical_ensemble} for the details. 
Therefore, when a basis is fixed, an approximation of a given classical state is equivalent to an approximation of a probability distribution corresponding to the classical state. 
To consider an approximation of a probability distribution, \textit{binning} is a useful method, where we divide a given probability distribution into several groups (bins) and replace each probability with the average value of the bin the probability belongs to. 

\begin{definition}[Binning]
Let $\{p_a:a\in\Sigma\}$ be a probability distribution. 
Consider a partition of an alphabet $\Sigma$ into $m$ disjoint constituent alphabets $\Sigma_k$, 
\begin{equation}\label{eq:bin_disjoint_union}
    \Sigma = \bigcup_{k=1}^m \Sigma_k \,,
\end{equation}
\begin{equation}
    \Sigma_k \cap \Sigma_l = \emptyset
\end{equation}
for all $k\neq l$. 
Then, we call each $\Sigma_k$ a \textit{bin} and $m$ the number of bins. 
Consider a probability distribution on $\Sigma$, $\{p_a:a\in\Sigma\}$.
For each bin $\Sigma_k$, let $p^{(k)}$ be the average value of the set $\{p_a:a\in\Sigma_k\}$. 
Then, we can construct a new probability distribution $\{p'_a:a\in\Sigma\}$ by replacing $p_a$ with $p'_a = p^{(k)}$ when $a\in\Sigma_k$. 
This method for generating a probability distribution $\{p'_a:a\in\Sigma\}$ from a given distribution $\{p_a:a\in\Sigma\}$ is called \textit{binning}.
\end{definition}
When probabilities within a bin $\{p_a:a\in\Sigma_k\}$ are all close to each other, binning gives an approximation of probability distributions; that is, it is also regarded as a good approximation method of classical states. 

We propose two methods to approximate a given ensemble based on binning. 
Here, we further specialize from a general two-state classical ensemble to 
one in which one of the two states is the flat state, motivated by the setup discussed in Ref.~\cite{Anshu2022}.

Suppose $\mathcal{H}$ is a quantum system, and consider two  classical states $\{\rho,\sigma\}$ where $\rho$ is the flat state. 
Take an orthonormal basis $\{\ket{i}\in\mathcal{H}:1\leq i \leq d_{\mathcal{H}}\}$ of $\mathcal{H}$ so that we can write
\begin{align}
    \rho &= \sum_{i=1}^{d_{\mathcal{H}}} \frac{1}{d_{\mathcal{H}}}\proj{i},\\
    \sigma &= \sum_{i=1}^{d_{\mathcal{H}}} p_i\proj{i},  
\end{align}
where $\{p_i : 1\leq i\leq d_{\mathcal{H}}\}$ forms a probability distribution with $p_1\geq p_2 \geq \cdots \geq p_{d_{\mathcal{H}}}$. 
We consider an approximation of $\{p_i : 1\leq i\leq d_{\mathcal{H}}\}$ within an allowed error $\epsilon>0$, 
that is, a probability distribution $\{p_i':1\leq i\leq d_{\mathcal{H}}\}$ such that 
\begin{equation}
    \label{eq:allowed_error}
    \sum_{i=1}^{d_{\mathcal{H}}} |p_i-p_i'| \leq \epsilon. 
\end{equation}
To construct an approximation leading to large redundant parts, we propose the following two methods, namely, the \textit{arithmetic mean method} and the \textit{geometric mean method}. 

\begin{itemize}
    \item Arithmetic Mean Method
    
    \begin{enumerate}
        \item First, for given $\epsilon>0$, find the largest positive integer $k_1$ such that 
        \begin{equation}
            \left|p_1 - p_{k_1}\right| \leq \frac{\epsilon}{d_{\mathcal{H}}}. 
        \end{equation}
        Define a set of positive integers $I_1 \coloneqq \{1,\ldots,k_1\}$. 
        \item For $i > 1$, find the largest positive integer $k_i$ such that
        \begin{equation}~\label{eq:arithmetic_mean_bin}
            \left|p_{k_{i-1}+1} - p_{k_{i}}\right| \leq \frac{\epsilon}{d_{\mathcal{H}}}.
        \end{equation}
        Define a set of positive integers $I_i \coloneqq \{k_{i-1}+1,\ldots,k_i\}$. 
        \item Repeat Step 2 until we find a positive integer $L$ such that $k_L = d_{\mathcal{H}}$. 
        \item For $j \in I_i$ with some $1\leq i \leq L$, replace $p_j$ as
        \begin{equation}
            p_j \mapsto p_j^{(A)} \coloneqq \frac{1}{|I_i|}\sum_{m \in I_i} p_m, 
        \end{equation}
        and define
        \begin{equation}
            \sigma^{(A)} \coloneqq \sum_{i=1}^{d_{\mathcal{H}}} p_i^{(A)} \proj{i}
        \end{equation}
    \end{enumerate}
    
    \item Geometric Mean Method
    
    \begin{enumerate}
        \item First, for given $\epsilon>0$, find the largest positive integer $k_1$ such that 
        \begin{equation}
            \frac{p_{k_1}}{p_1} \geq \frac{1}{1+\epsilon}. 
        \end{equation}
        Define a set of positive integers $I_1 \coloneqq \{1,\ldots,k_1\}$. 
        \item For $i>1$, find the largest positive integer $k_i$ such that
        \begin{equation}~\label{eq:geometric_mean_bin}
            \frac{p_{k_i}}{p_{k_{i-1}+1}} \geq \frac{1}{1+\epsilon}. 
        \end{equation}
        Define a set of positive integers $I_i \coloneqq \{k_{i-1}+1,\ldots,k_i\}$. 
        \item Repeat Step 2 until we find a positive integer $L$ such that $k_L = d_{\mathcal{H}}$. 
        \item For $j \in I_i$ with some $1\leq i \leq l$, replace $p_j$ as
        \begin{equation}
            p_j \mapsto p_j^{(G)} \coloneqq \frac{1}{|I_i|}\sum_{m \in I_i} p_m, 
        \end{equation}
        and define
        \begin{equation}
            \sigma^{(G)} \coloneqq \sum_{i=1}^{d_{\mathcal{H}}} p_i^{(G)}\proj{i}. 
        \end{equation}
    \end{enumerate}
\end{itemize}

We show that these approximation methods are valid; that is, the resulting states $\sigma^{(A)}$ and $\sigma^{(G)}$ are close enough to the original state $\sigma$. 
The following proposition guarantees that $\{\rho,\sigma'\}$ with $\sigma'=\sigma^{(A)}\,\mathrm{or}\,\sigma^{(G)}$ is considered as an approximation of $\{\rho,\sigma\}$. 
\begin{proposition}
For a positive real number $\epsilon > 0$, let $\sigma^{(A)}$ and $\sigma^{(G)}$ be the approximate states of a given diagonal state $\sigma$ by the arithmetic mean method and the geometric mean method respectively. 
Then, it holds that 
\begin{align}
    \|\sigma^{(A)} - \sigma\|_1 &\leq \epsilon,\\
    \|\sigma^{(G)} - \sigma\|_1 &\leq \epsilon. 
\end{align}
\end{proposition}
\begin{proof}
It holds that
\begin{equation}
    \begin{aligned}
    \|\sigma - \sigma^{(A)}\|_1
    &= \sum_{i = 1}^{d_{\mathcal{H}}} |p_i - p_i^{(A)}|\\
    &= \sum_{i = 1}^{L}\sum_{j \in I_i} |p_j - p_j^{(A)}|\\
    &\leq \sum_{i = 1}^{L}\sum_{j \in I_i} |p_{k_{i-1}+1} - p_{k_i}|\\
    &\leq \sum_{i = 1}^{L}\sum_{j \in I_i} \frac{\epsilon}{d_{\mathcal{H}}}\\
    &= \epsilon. 
    \end{aligned}
\end{equation}
The first inequality follows because the difference $|p_j - p_j^{(A)}|$ is upper-bounded by the difference between the largest value and the smallest value of the set $\{p_i : i\in I_j\}$, which is given by $|p_{k_{i-1}+1} - p_{k_i}|$.
The second inequality follows from Eq.~\eqref{eq:arithmetic_mean_bin}.

On the other hand, it holds that
\begin{equation}~\label{eq:geometric_1}
    \begin{aligned}
     \|\sigma - \sigma^{(G)}\|_1
    &= \sum_{i = 1}^{d_{\mathcal{H}}} |p_i - p_i^{(G)}|\\
    &= \sum_{i = 1}^{L}\sum_{j \in I_i} |p_j - p_j^{(G)}|\\
    &= \sum_{i = 1}^{L}\sum_{j \in I_i} p_j^{(G)}\left|\frac{p_j}{p_j^{(G)}} - 1\right|.
    \end{aligned}
\end{equation}
Here, by Eq.~\eqref{eq:geometric_mean_bin}, we have
\begin{equation}
     \frac{1}{1+\epsilon}\leq \frac{p_j}{p_j^{(G)}} \leq 1+\epsilon, 
\end{equation}
which implies that 
\begin{equation}~\label{eq:geometric_2}
     -\epsilon \leq \frac{-\epsilon}{1+\epsilon}\leq \frac{p_j}{p_j^{(G)}} - 1 \leq \epsilon, 
\end{equation}
Combining Eqs.~\eqref{eq:geometric_1} and \eqref{eq:geometric_2}, we have that 
\begin{equation}
    \begin{aligned}
     \|\sigma - \sigma^{(G)}\|_1
    &= \sum_{i = 1}^{L}\sum_{j \in I_i} p_j^{(G)}\left|\frac{p_j}{p_j^{(G)}} - 1\right|\\
    &\leq \sum_{i = 1}^{L}\sum_{j \in I_i} p_j^{(G)} \epsilon\\
    &= \epsilon.  
    \end{aligned}
\end{equation}
\end{proof}

Consider a quantum ensemble formed by $\{\rho,\sigma'\}$ ($\sigma'=\sigma^{(A)}\,\mathrm{or}\,\sigma^{(G)}$) and quantum channels $\mathcal{K}_{\mathrm{off}}$ and $\mathcal{K}_{\mathrm{on}}$ with respect to this ensemble. 
Since
\begin{align}
    \mathcal{K}_{\mathrm{off}}(\rho) 
    &= \sum_{i=1}^{L} \frac{|I_i|}{d_{\mathcal{H}}}\proj{i},\\
    \mathcal{K}_{\mathrm{off}}(\sigma) 
    &= \sum_{i=1}^{L} \left(\sum_{m\in I_i}p_m\right)\proj{i}, 
\end{align}
the rate $R$ of our protocol for this ensemble is
\begin{equation}~\label{eq:classical_upperbound}
    R = \entropy\left(p_{\rho}\mathcal{K}_{\mathrm{off}}(\rho) +p_{\sigma}\mathcal{K}_{\mathrm{off}}(\sigma) \right)
    \leq \log_2L. 
\end{equation}
In addition, by construction, 
\begin{align}
     \mathcal{K}_{\mathrm{on}}\circ \mathcal{K}_{\mathrm{off}}(\rho) 
    &= \rho,\\
     \mathcal{K}_{\mathrm{on}}\circ \mathcal{K}_{\mathrm{off}}(\sigma) 
    &= \sigma'.  
\end{align}
Thus, this approximation satisfies the condition~\eqref{eq:condition_approx_protocol}, namely, 
\begin{align}
    \| \mathcal{K}_{\mathrm{on}}\circ \mathcal{K}_{\mathrm{off}}(\rho) - \rho\|_1 &= \|\rho - \rho\|_1  = 0 \leq \epsilon, \\
    \| \mathcal{K}_{\mathrm{on}}\circ \mathcal{K}_{\mathrm{off}}(\sigma) - \sigma\|_1 &= \|\sigma' - \sigma\|_1 \leq \epsilon. 
\end{align}
While analytical discussion of these two methods can be complicated, 
in the following subsection, 
we conduct numerical experiments to investigate 
the performance of these approximation methods.

\subsection{Numerical Simulation for Approximating Classical Ensemble}
In this section, we perform numerical simulations to verify and compare the performance of the two methods, the \textit{arithmetic mean method} and the \textit{geometric mean method}. 
In Ref.~\cite{Anshu2022}, the authors proposed a protocol for approximately compressing a two-state classical ensemble, 
which employs the geometric mean. 
One may wonder if the arithmetic mean yields better results. 

We conduct two types of simulations. 
In the first one, we fix the dimension of the quantum system and vary 
the allowed approximation $\epsilon$. 
Then, we observe how the compression rates obtained by these approximation methods depend on $\epsilon$. 
In the second one, we change the dimension of quantum system $d_{\mathcal{H}}$, and we set $\epsilon = 1/\sqrt{d_{\mathcal{H}}}$ for each $d_{\mathcal{H}}$. 
Then, we see the dependence of the compression rates on the dimension of a given system. 

\begin{figure*}[htbp]
        \centering
        \includegraphics[keepaspectratio, width=\linewidth]{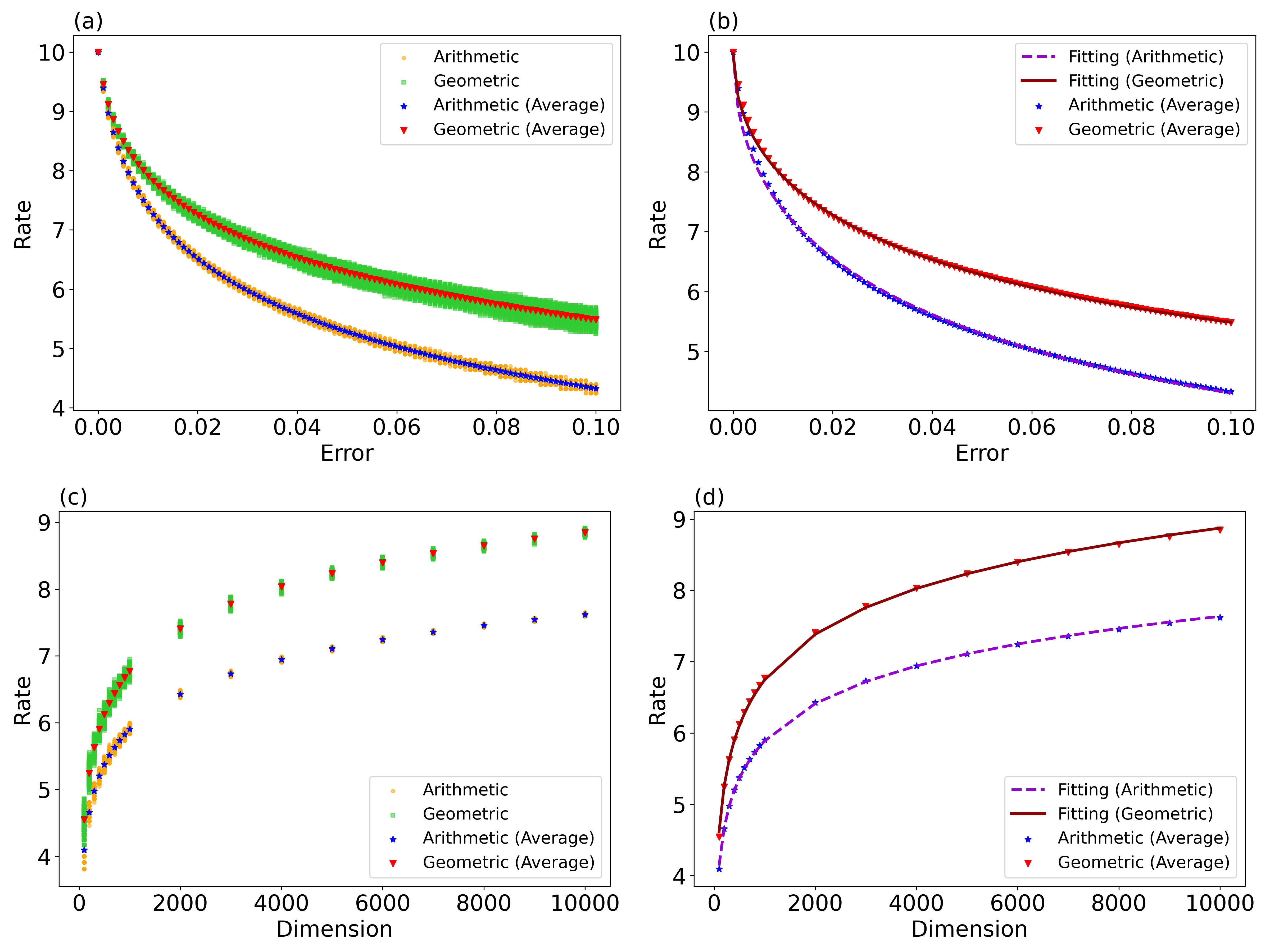}
        \caption{Results of numerical experiments. 
        (a): Graph of the compression rate as a function of allowed error $\epsilon$ shown in Eq.~\eqref{eq:allowed_error}. 
        (b): Graph of the fitting curves of the average results of (a).
        (c): Graph of the compression rate~\eqref{eq:classical_upperbound} as a function of the dimension $D$ of the system under allowed approximation $1/\sqrt{D}$. 
        (d): Graph of the fitting curves of the average results of (c).
        In (a) and (c), the yellow circles and green squares represent the compression rates of randomly generated probability distributions obtained by the arithmetic and geometric mean methods, respectively. 
        The blue stars and red triangles represent the averages of the yellow circles and green squares, respectively; 
        that is, the blue stars show the averaged results of the arithmetic mean methods, and the red triangles show those of the geometric mean methods. 
        In (b) and (d), the violet dashed lines are the fitting of the blue stars; the brown solid lines are the fitting of the red triangles.
        In (b), we employ the model function defined in Eq.~\eqref{eq:fitting_rate_error}, and we obtained the fitting parameters $a = 7.859(6)$, $b = 4.077(5)$ for the violet dashed line (the arithmetic mean method) and $a = 6.262(2)$, $b = 4.030(2)$ for the brown solid line (the geometric mean method). 
        In (d), we use the function defined in Eq.~\eqref{eq:fitting_rate_dim}., 
        and we obtained the fitting parameters $a = 0.525892(4)$, $b = 0.6477(4)$ for the violet dashed line (the arithmetic mean method) and $a = 0.640847(9)$, $b = 0.358(1)$ for the brown solid line (the geometric mean method).
        The term ``rate'' in the $y$-axes refers to the upper bound shown in Eq.~\eqref{eq:classical_upperbound}. 
        Note that the quantities shown in the graphs are dimensionless. 
        }
        \label{fig:error_rate}
\end{figure*}

\subsubsection{Numerical Simulation for Approximation-Rate Trade-off}
    Here, we evaluate the performance of the two approximation methods with various values of allowed approximation $\epsilon$. 
    In the simulation, we fix the space to $\mathcal{H} = \mathbb{C}^{1024}$; that is, $d_{\mathcal{H}} = 1024 = 2^{10}$. 
    We generate 1000 random diagonal states $\sigma$ on this space. 
    In more detail, we randomly generate a real number in the range $[0,1]$ according to the uniform distribution $1024$ times to obtain a vector on $\mathbb{R}^{1024}$. 
    Then, we normalize the obtained vector to obtain a probability distribution, and we regard this probability distribution as a classical state. 
    For each state, we create the approximate states $\sigma^{(A)}$ and $\sigma^{(G)}$ by using 
    the arithmetic mean method and the geometric mean method. 
    
    We compare the rates of the two methods. 
    Figure~\ref{fig:error_rate}(a) shows that the arithmetic mean method performs better than the geometric mean method on average. 
    In particular, the results show the tendency that the difference between the two methods becomes large as the allowed error increases. 
    
    Furthermore, we estimate the curves representing average values of the results in the figure.
    Observing the graph shown in Figure~\ref{fig:error_rate}(a), we adopt the following function
    \begin{equation}~\label{eq:fitting_rate_error}
        f(x) = \log_2 d_{\mathcal{H}} - a\left(1-\mathrm{e}^{-b\sqrt{x}}\right)
    \end{equation}
    with parameters $a$ and $b$ as a fitting function. 
    Since the compression rate is strictly upper-bounded by $\log_2 d_{\mathcal{H}}$ when we do not allow any approximation, the second term $a\left(1-\mathrm{e}^{-b\sqrt{x}}\right)$ of Eq.~\eqref{eq:fitting_rate_error} is considered to represent the degree of reduction caused by an approximation. 
    In addition, this function explains the tendency that the rate approaches $\log_2 d_{\mathcal{H}}$ as the error becomes small, and it also expresses the sudden decrease around the error $1/d_{\mathcal{H}}$. 
    The results shown in Figure~\ref{fig:error_rate}(b) suggest that the fitting function defined in Eq.~\eqref{eq:fitting_rate_error} is the correct function characterizing the compression rate as a function of the error
    while we have not theoretically and analytically demonstrated it. 
    However, we do not believe that this function explains the full dependency of the rate on the error, because of the sensitive behavior of the rate against errors when errors are roughly larger than $1/\sqrt{d_{\mathcal{H}}}$. 
    For more general and deep understandings of the dependency, we need further investigations of the compression protocols. 
    
    \subsubsection{Numerical Simulation for Dimension-Rate Relation}
    We investigate the dependence of the performance of the two approximation methods on the dimension of the space. 
    Here, we adopt the error $\epsilon = 1/\sqrt{d_{\mathcal{H}}}$ 
    considering the discussion in Ref.~\cite{Anshu2022} stating that this size of error $f(\Lambda_n)\approx 1/\sqrt{d_{\mathcal{H}}}$ (see Eq.~\eqref{eq:error_f}) leads to 
    $g(\Lambda_n) \approx \log_2 d$ (see Eq.~\eqref{eq:error_g}). 
    See Appendix A of Ref.~\cite{Anshu2022} for more details. 
    Thus, $\epsilon = 1/\sqrt{d_{\mathcal{H}}}$ may well cause a large reduction of the rate. 
    We generate 1000 random diagonal states $\sigma$ on this space in the same way as in the first simulation. 
    For each state, we create the approximate 
    states $\sigma^{(A)}$ and $\sigma^{(G)}$ by using the arithmetic mean method and the geometric mean method. 
    
    Then, we compare the rates of the two methods. 
    Figure~\ref{fig:error_rate}(c) shows that the arithmetic mean method also performs better than the geometric mean method on average in this case. 
    
    To investigate the dependence of the difference between the two methods on the dimension, we also plotted the difference in Figure~\ref{fig:division}.
    As seen in the graph, the difference becomes large as the dimension gets large, and it linearly depends on the logarithm of dimension. 
    
    Moreover, we estimate the curves representing average values of the results in the graph.
    Observing the graphs shown in Figures~\ref{fig:error_rate}(c) and \ref{fig:division},
    adopt the following function
    \begin{equation}~\label{eq:fitting_rate_dim}
        f(x) = a\log_2x + b
    \end{equation}
    with parameters $a$ and $b$ as a fitting function.  
    The results shown in Figure~\ref{fig:error_rate}(d) imply that the fitting function defined in Eq.~\eqref{eq:fitting_rate_dim} is appropriate for characterizing the compression rate as a function of the dimension.  
    Since the compression rate is strictly upper-bounded by $\log_2 d_{\mathcal{H}}$ when we do not allow any approximation, the coefficient $a$ represents the degree of reduction caused by an approximation. 
 
    \begin{figure}
    \centering
    \includegraphics[keepaspectratio, width=\linewidth]{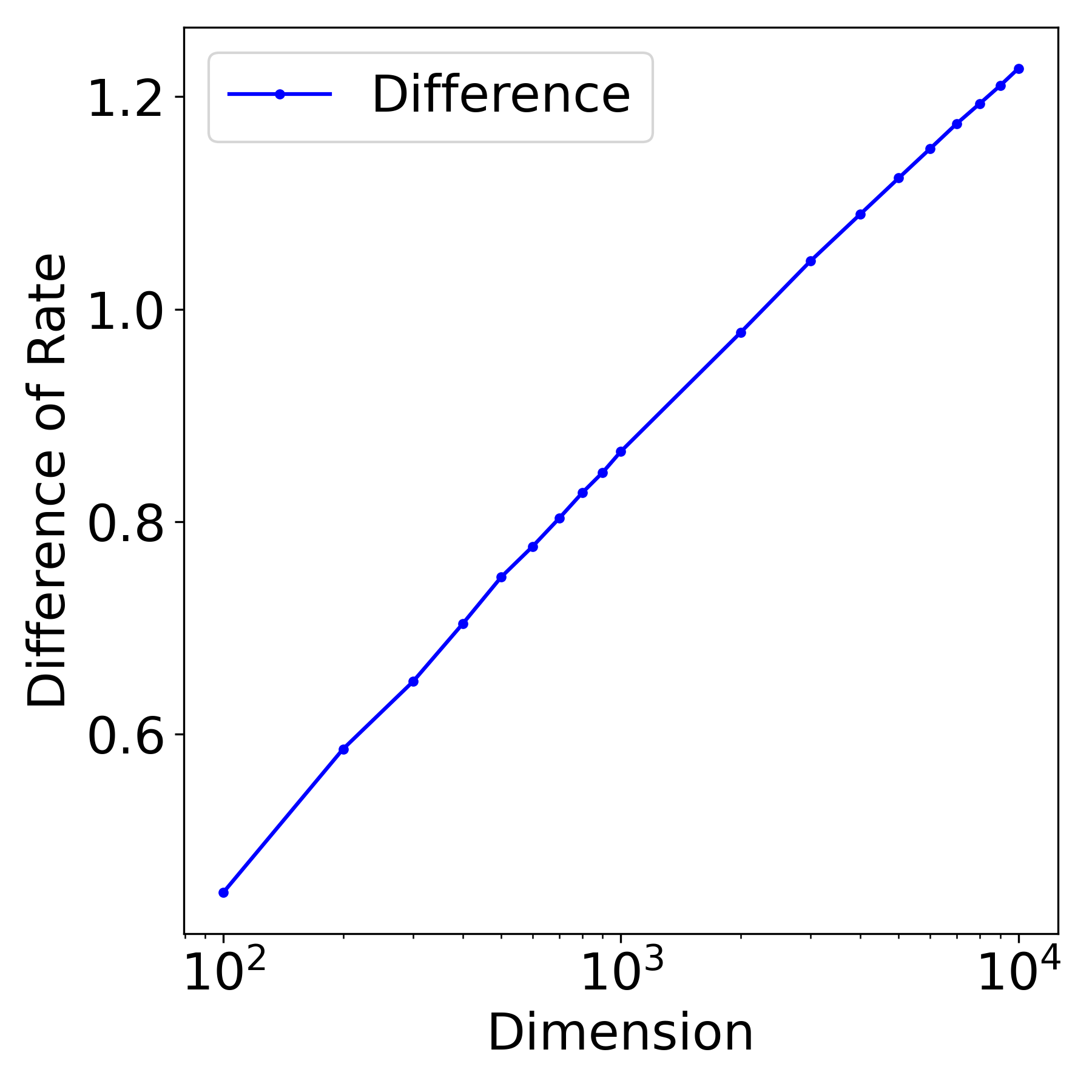}
    \caption{Graph of (the compression rate by the arithmetic mean method) $-$ (the compression rate by the geometric mean method) as a function of the dimension of system. 
    Here, ``rate'' refers to the upper bound shown in Eq.~\eqref{eq:classical_upperbound}.
    The horizontal axis is expressed in log scale.
    Note that the quantities shown in the graphs are dimensionless. 
    }
    \label{fig:division}
    \end{figure}

\section{Discussion and Conclusion}~\label{sec:Blind_discussions}
In this paper, we investigated blind compression of quantum ensembles under finite local approximations. 

In previous research, the optimal rate of blind compression was obtained through the KI decomposition 
by removing the redundant parts of a given ensemble. 
However, since KI decomposition is highly sensitive even to small approximations, 
approximation or error analysis of blind compression has not been explored much.
In this paper, we focused on the instability of the KI decomposition against approximations.
Taking advantage of the sensitivity of KI decomposition, we constructed a compression protocol that can have a substantial reduction of the compression rate if a finite approximation is allowed. 
In our protocol, the sender intentionally introduces distortion to a given ensemble by applying KI operations that almost preserve the ensemble. 
Let us note that the distortion implemented by the sender plays a major role in a full quantum rate-distortion theory for mixed states~\cite{Khanian_Kuroiwa2022}. 
We explicitly showed a reduction of the compression rate by several examples. 
In particular, in Example~\ref{Example2_blindcomp}, we saw a near-maximal reduction of the compression rate allowed by the dimension. 
Remarkably, the rate reduction of our protocol persists even when approximation $\epsilon$ is $\approx 1/d_{\mathcal{H}}$, contrasting to a result of Ref.~\cite{Anshu2022} for two classical distributions, where a large rate reduction is seen for $\epsilon \approx 1/\sqrt{d_{\mathcal{H}}}$. 
We also share the following insightful comment from one of our referees. 
For any quantum ensemble, $\Phi_1$, for any neighborhood of $\Phi_1$, there is an ensemble $\Phi_2$ in this neighborhood with trivial KI decomposition. 
If the initial ensemble $\Phi_1$ has redundant parts, the two ensembles $\Phi_1$ and $\Phi_2$ will have different compression rates. Our examples can be understood as a manifestation of this generic situation. 
We want to point out that our simple explicit example with only two states provides a nearly extremal quantitative demonstration. 
Moreover, we investigated blind compression of classical ensembles to analyze general properties of our compression protocol. 
We proposed two compression methods for two-state classical ensembles including the flat state, namely, the arithmetic mean method and the geometric mean method. 
We numerically investigated these compression methods of two-state classical ensembles and discovered that the arithmetic mean method performs better than the geometric mean method. 

As a future direction, it would be interesting to analyze our protocol more generally. 
As noted in the previous sections, our protocol does not necessarily show a large reduction of the compression rate for an arbitrary ensemble.  
For example, if the allowed approximation is sufficiently small, our protocol might have rate equal to the optimal compression rate with a vanishing error, obtained in Ref.~\cite{Koashi2001}.
The difficulty of the analysis lies in the absence of an approximation of the KI decomposition. 
While finite approximations can significantly lower the compression rate as we showed in this chapter because of the instability of the KI decomposition, the instability also makes the general analysis intractable. 
Investigation of compression protocols on general two-state quantum ensembles might be a good first step for this direction. 
For a two-state ensemble, we may consider that one of the two states is diagonal by choosing an appropriate basis. 
Fixing the basis would be helpful to study a new approximate compression protocol and the dependence of the compression protocol on allowed errors. 

Thus, we shed light on blind quantum data compression with the allowance of approximations 
by revealing that the approximation-sensitive nature of KI decomposition can be used to achieve a substantial rate reduction. 
We believe that our findings fuel further investigation of blind compression with finite approximations, leading to a more general understanding of approximation-rate trade-off in this setup. 

\begin{acknowledgments}
We thank Anurag Anshu, Felix Leditzky, and Shima Bab Hadiashar for helpful discussions. 
We are grateful to Crystal Senko, John Watrous, and Norbert L\"{u}tkenhaus for their advice and feedback.
KK was supported by a Mike and Ophelia Lazaridis Fellowship, the Funai Foundation, and a Perimeter Residency Doctoral Award. 
DL was supported by NSERC. 
\end{acknowledgments}

\bibliography{BlindCompression}

\clearpage
\appendix
\section{Proof of Theorem~\ref{thm:KIoperations}}~\label{proof_KIoperation}
    We can give the descriptions of $\mathcal{K}_{\mathrm{off}}$ and $\mathcal{K}_{\mathrm{on}}$ explicitly by using Kraus representations. 
    First, $\mathcal{K}_{\mathrm{off}}$ is given by Kraus operators
    \begin{equation}
        A^{(l)}_{j_l} \coloneqq I_{\mathcal{H}^{(l)}_Q} \otimes \bra{j_l},
    \end{equation}
    where $\{\ket{j_l}:j_l\}$ is an orthonormal basis of $\mathcal{H}^{(l)}_Q$ for all $l \in \Xi$.
    Kraus operators $A^{(l)}_{j_l}$ apply to the $l$th block, and $\{A^{(l)}_{j_l}:j_l\}$ forms the partial trace over $\mathcal{H}^{(l)}_R$. 
    Indeed, 
    \begin{equation}
        \begin{aligned}
        \sum_{l\in\Xi} \sum_{j_l} A^{(l)}_{j_l} \rho_a (A^{(l)}_{j_l} )^\dagger
        &= \bigoplus_{l \in \Xi} q^{(a,l)}\rho^{(a,l)}_Q\otimes \left(\sum_{j_l}\braket{j_l|\rho^{(l)}_R|j_l}\right)\\
        &= \bigoplus_{l \in \Xi} q^{(a,l)}\rho^{(a,l)}_Q \otimes \tr(\rho^{(l)}_R)\\
        &=\bigoplus_{l \in \Xi} q^{(a,l)}\rho^{(a,l)}_Q,
        \end{aligned}
    \end{equation}
    for all $a\in\Sigma$. 
    In addition, $ \{A^{(l)}_{j_l}:l\in\Xi, j_l\}$ forms a valid Kraus representation because 
    \begin{equation}
        \begin{aligned}
            \sum_{l\in\Xi} \sum_{j_l}(A^{(l)}_{j_l} )^\dagger A^{(l)}_{j_l} 
            &= \sum_{l\in\Xi} \sum_{j_l} I_{\mathcal{H}^{(l)}_Q} \otimes \proj{j_l}\\
            &= \sum_{l\in\Xi} I_{\mathcal{H}^{(l)}_Q} \otimes \left(\sum_{j_l}\proj{j_l}\right)\\
            &= \sum_{l\in\Xi} I_{\mathcal{H}^{(l)}_Q} \otimes I_{\mathcal{H}^{(l)}_R}\\
            &= I_{\mathcal{H}}.
        \end{aligned}
    \end{equation}
    Therefore, we can construct $\mathcal{K}_{\mathrm{off}}$ by a Kraus representation
    \begin{equation}
        \{A^{(l)}_{j_l}:l\in\Xi, j_l\}.
    \end{equation}
    Next, $\mathcal{K}_{\mathrm{on}}$ is given by Kraus operators
    \begin{equation}
        A^{(l)}_{k_l} \coloneqq I_{\mathcal{H}^{(l)}_Q} \otimes \sqrt{r_{k_l}}\ket{k_l},
    \end{equation}
    where $\{\ket{k_l}:k_l\}$ is an orthonormal basis of $\mathcal{H}^{(l)}_Q$ for all $l \in \Xi$ corresponding to a spectral decomposition 
    \begin{equation}
        \rho^{(l)}_R \coloneqq \sum_{k_l} r_{k_l}\proj{k_l}
    \end{equation}
    with eigenvalues $r_{k_l} \geq 0$.
    Kraus operators $A^{(l)}_{k_l}$ apply to the $l$th block, and $\{A^{(l)}_{k_l}:k_l\}$ forms the construction of $\rho^{(l)}_R$ on system $\mathcal{H}^{(l)}_R$. 
    Indeed, we have that
    \begin{equation}
        \begin{aligned}
        &\sum_{l\in\Xi} \sum_{k_l} A^{(l)}_{k_l} \left(\bigoplus_{l \in \Xi} q^{(a,l)}\rho^{(a,l)}_Q\right) (A^{(l)}_{k_l} )^\dagger\\
        &= \bigoplus_{l \in \Xi} q^{(a,l)}\rho^{(a,l)}_Q\otimes \left(\sum_{k_l} r_{k_l}\proj{k_l}\right)\\
        &= \bigoplus_{l \in \Xi} q^{(a,l)}\rho^{(a,l)}_Q \otimes \rho^{(l)}_R\\
        &=\rho_a
        \end{aligned}
    \end{equation}
    for all $a\in\Sigma$. 
    In addition, $ \{A^{(l)}_{k_l}:l\in\Xi, k_l\}$ forms a valid Kraus representation because it holds that
    \begin{equation}
        \begin{aligned}
            \sum_{l\in\Xi} \sum_{k_l}(A^{(l)}_{k_l} )^\dagger A^{(l)}_{k_l} 
            &= \sum_{l\in\Xi} \sum_{k_l} I_{\mathcal{H}^{(l)}_Q} \otimes r_{k_l}\braket{k_l|k_l}\\
            &= \sum_{l\in\Xi} I_{\mathcal{H}^{(l)}_Q} \otimes \left(\sum_{k_l}r_{k_l}\right)\\
            &= \sum_{l\in\Xi} I_{\mathcal{H}^{(l)}_Q} \otimes 1\\
            &= I_{\mathcal{H}^{(l)}_Q}.
        \end{aligned}
    \end{equation}
    Therefore, we can construct $\mathcal{K}_{\mathrm{on}}$ by  a Kraus representation
    \begin{equation}
        \{A^{(l)}_{k_l}:l\in\Xi, k_l\}.
    \end{equation}

\setcounter{theorem}{0}
\renewcommand{\thetheorem}{B\arabic{theorem}}
\section{Approximation of a quantum ensemble in a fixed basis}~\label{appendix:classical_ensemble}
A classical ensemble is defined with a fixed basis of a given quantum system in which every state in the ensemble is diagonalized. (See Definition~\ref{def:classical_ensemble}.)
Here, we analyze an approximation of a given classical ensemble in the fixed basis; that is, we consider an approximate ensemble whose structure of the KI decomposition is given in the same fixed basis. When we approximate a classical ensemble with respect to the fixed basis, we only have to consider an approximation of diagonal elements in the sense of the following proposition. 
\begin{proposition}
Let $\mathcal{H}$ be a quantum system. 
Let $\rho,\sigma \in \Density(\mathcal{H})$ be quantum states on the system. 
Suppose that $\rho$ and $\sigma$ form a classical ensemble; that is, we can write
\begin{align}
    \rho &= \sum_{i=1}^{d_{\mathcal{H}}} \rho_i \proj{i},\\
    \sigma &= \sum_{i=1}^{d_{\mathcal{H}}} \sigma_i \proj{i}
\end{align}
for an orthonormal basis $\{\ket{i}\in\mathcal{H}:1\leq i \leq d_{\mathcal{H}}\}$ of $\mathcal{H}$ and probability distributions $\{\rho_i: 1\leq i \leq d_{\mathcal{H}}\}$ and $\{\sigma_i: 1\leq i \leq d_{\mathcal{H}}\}$. 
Suppose that $\rho$ and $\sigma$ can be approximated in the same basis as 
\begin{align}
    &\left\|\rho - \bigoplus_{l\in\Xi} q^{(\rho,l)} \omega_Q^{(\rho,l)} \otimes \omega_R^{(l)} \right\|_1 \leq \epsilon,\\
    &\left\|\sigma - \bigoplus_{l\in\Xi} q^{(\sigma,l)} \omega_Q^{(\sigma,l)} \otimes \omega_R^{(l)} \right\|_1 \leq \epsilon
\end{align}
where the structure of the decomposition is also given in the basis $\{\ket{i}\in\mathcal{H}:1\leq i \leq d_{\mathcal{H}}\}$. Then, there exists an approximation such that all $\omega_Q^{(\rho,l)}$, $\omega_Q^{(\sigma,l)}$, and $\omega_R^{(l)}$ are diagonal. 
\end{proposition}
\begin{proof}
Let us take some $l \in \Xi$, and consider $q^{(\rho,l)} \omega_Q^{(\rho,l)} \otimes \omega_R^{(l)}$ and 
$q^{(\sigma,l)} \omega_Q^{(\sigma,l)} \otimes \omega_R^{(l)}$. Let $\rho^{(l)}$ and $\sigma^{(l)}$ be the corresponding block of $\rho$ and $\sigma$ respectively. 
Let $\Delta \in \Channel(\mathcal{H},\mathcal{H})$ be \textit{the completely dephasing channel} with respect to the basis $\{\ket{i}\in\mathcal{H}:1\leq i \leq d_{\mathcal{H}}\}$
defined as 
\begin{equation}
    \Delta(\rho) = \sum_{i = 1}^{d_{\mathcal{H}}} \braket{i|\rho|i} \proj{i}
\end{equation}
for $\rho \in \Density(\mathcal{H})$. 
Then, we have
\begin{equation}
    \begin{aligned}
    &\left\|\rho^{(l)} - \bigoplus_{l\in\Xi} q^{(\rho,l)} \Delta(\omega_Q^{(\rho,l)}) \otimes \Delta(\omega_R^{(l)})\right\|_1\\
    &= \left\|\Delta(\rho^{(l)}) - \Delta\left(\bigoplus_{l\in\Xi} q^{(\rho,l)} \omega_Q^{(\rho,l)} \otimes \omega_R^{(l)}\right)\right\|_1\\
    &\leq \left\|\rho^{(l)} - \bigoplus_{l\in\Xi} q^{(\rho,l)} \omega_Q^{(\rho,l)} \otimes \omega_R^{(l)}\right\|_1. 
    \end{aligned}
\end{equation}
Similarly, 
\begin{equation}
    \begin{aligned}
    &\left\|\sigma^{(l)} - \bigoplus_{l\in\Xi} q^{(\sigma,l)} \Delta(\omega_Q^{(\sigma,l)}) \otimes \Delta(\omega_R^{(l)})\right\|_1\\
    &\leq \left\|\sigma^{(l)} - \bigoplus_{l\in\Xi} q^{(\sigma,l)} \omega_Q^{(\sigma,l)} \otimes \omega_R^{(l)}\right\|_1. 
    \end{aligned}
\end{equation}
Therefore, diagonal states $\bigoplus_{l\in\Xi} q^{(\rho,l)} \Delta(\omega_Q^{(\rho,l)}) \otimes \Delta(\omega_R^{(l)})$ and $\bigoplus_{l\in\Xi} q^{(\sigma,l)} \Delta(\omega_Q^{(\sigma,l)}) \otimes \Delta(\omega_R^{(l)})$ are also approximations of $\rho$ and $\sigma$. 
\end{proof}

\setcounter{theorem}{0}
\renewcommand{\thetheorem}{C\arabic{theorem}}
\section{Proof that $\sigma_1$ and $\sigma_2$ in Example~\ref{Example2_blindcomp} have no redundant parts}~\label{appendix:example2}
In this section, we prove that the two $2N$-dimensional states
\begin{align*}
   \sigma_1 &\coloneqq 
    \frac{1}{4N}\left(
    \begin{array}{ccccc}
    2 & \epsilon \mathrm{e}^{i\alpha} & 0 & \cdots & \epsilon \mathrm{e}^{-i\alpha} \\
    \epsilon\mathrm{e}^{-i\alpha} & 2 & \epsilon\mathrm{e}^{i\alpha} & \cdots & 0 \\
    \vdots &\ddots & \ddots & \ddots & \vdots\\
    0 & \cdots & \epsilon\mathrm{e}^{-i\alpha}& 2 & \epsilon\mathrm{e}^{i\alpha} \\
    \epsilon\mathrm{e}^{i\alpha} & 0 & \cdots & \epsilon\mathrm{e}^{-i\alpha} & 2 
    \end{array}
    \right)
    \\
    \sigma_2 &\coloneqq 
    \frac{1}{4N}\left(
    \begin{array}{cccccc}
    1+2\epsilon &        &   &   &        &   \\
      & \ddots &   &   &        &   \\
      &        & 1 + 2\epsilon&   &        &   \\
      &        &   & 3 -2\epsilon&        &   \\
      &        &   &   & \ddots &   \\
      &        &   &   &        & 3-2\epsilon  
    \end{array}
    \right)
\end{align*}
with $0<\epsilon<1/2$ and $0 < \alpha < 1/(4N)$, introduced in Example~\ref{Example2_blindcomp}, do not have redundant parts. 

First, we prove that $\sigma_1$ and $\sigma_2$ do not have a non-trivial (multi-block) block-diagonal structure. 
Define normalized vectors 
\begin{equation}
    \ket{v_k} = 
    \frac{1}{\sqrt{2N}}\left(
    \begin{array}{c}
        1   \\
        \omega^{k} \\ 
        \omega^{2k} \\ 
        \omega^{3k} \\ 
        \vdots \\ 
        \omega^{(2N-1)k}
    \end{array}
    \right)
\end{equation}
with
\begin{equation}
    \omega = \mathrm{e}^{i\tfrac{2\pi}{2N}}
\end{equation}
for $k = 0,1,\ldots, 2N-1$.
Note that $\omega^{2N} = 1$. 
Then, 
\begin{equation}
    \sigma_1 \ket{v_k} = \frac{1 + \epsilon\left(\cos\left(\frac{2\pi k}{2N} + \alpha \right)\right)}{2N} \ket{v_k}
\end{equation}
for $k = 0,1,\ldots,2N-1$; 
that is, $\ket{v_k}$ is an eigenvector of $\sigma_1$ with eigenvalue $\left(\cos\left(\frac{2\pi k}{2N} + \alpha \right)\right)/(2N)$.
Since $0 < \alpha < 1/(4N)$, 
there do not exist distinct $k,l \in\{ 0,1,2,\ldots,2N-1\}$ such that 
\begin{equation}
    \cos\left(\frac{2\pi k}{2N} + \alpha \right) = \cos\left(\frac{2\pi l}{2N} + \alpha \right). 
\end{equation}
That is, $\sigma_1$ is not degenerate. 
On the other hand, 
eigenvectors of $\sigma_2$ can be expressed as 
\begin{equation}
    \left(
    \begin{array}{c}
        a_1   \\
        a_2 \\ 
        \vdots  \\ 
        a_{N} \\ 
        0 \\
        0 \\
        \vdots \\ 
        0
    \end{array}
    \right)
    \,\, 
    \mathrm{or}
    \,\,
    \left(
    \begin{array}{c}
        0 \\
        0 \\
        \vdots \\ 
        0 \\ 
        a_1   \\
        a_2 \\ 
        \vdots  \\ 
        a_{N} \\ 
    \end{array}
    \right), 
\end{equation}
where $(a_{1},a_{2},\ldots,a_{N}) \neq (0,0,\ldots,0)$. 
Without loss of generality, we focus on normalized vectors $(a_{1},a_{2},\ldots,a_{N})$. 
Indeed, 
\begin{equation}
    \sigma_2
    \left(
    \begin{array}{c}
        a_1   \\
        a_2 \\ 
        \vdots  \\ 
        a_{N} \\ 
        0 \\
        0 \\
        \vdots \\ 
        0
    \end{array}
    \right)
    = \frac{1 + 2\epsilon}{4N} \left(
    \begin{array}{c}
        a_1   \\
        a_2 \\ 
        \vdots  \\ 
        a_{N} \\ 
        0 \\
        0 \\
        \vdots \\ 
        0
    \end{array}
    \right)
\end{equation}
and 
\begin{equation}
    \sigma_2
    \left(
    \begin{array}{c}
        0 \\
        0 \\
        \vdots \\ 
        0 \\ 
        a_1   \\
        a_2 \\ 
        \vdots  \\ 
        a_{N} \\ 
    \end{array}
    \right)
    = \frac{3 - 2\epsilon}{4N} 
    \left(
    \begin{array}{c}
        0 \\
        0 \\
        \vdots \\ 
        0 \\ 
        a_1   \\
        a_2 \\ 
        \vdots  \\ 
        a_{N} \\ 
    \end{array}
    \right). 
\end{equation}
Note that $1+2\epsilon < 3 -2\epsilon$ since $0 < \epsilon < 1/2$. 

Suppose that $\sigma_1$ and $\sigma_2$ have a shared non-trivial block diagonal structure. 
In this case, there exist a unitary matrix $U$, non-trivial subspaces $\tilde{\mathcal{H}},\mathcal{H}'$ with $d_{\tilde{\mathcal{H}}} + d_{\mathcal{H}}' = 2N$, and positive semidefinite matrices $P_1,P_2$ on $\tilde{\mathcal{H}}$ and $Q_1,Q_2$ on $\mathcal{H}'$ such that 
\begin{align}
    \label{eq:sigma_1_block}
    U\sigma_1 U^\dagger &= P_1\oplus Q_1, \\ 
    \label{eq:sigma_2_block}
    U\sigma_2 U^\dagger &= P_2\oplus Q_2. 
\end{align}
Since $\sigma_1$ is not degenerate, 
from Eq.~\eqref{eq:sigma_1_block}, 
vector space $\tilde{\mathcal{H}}$ is spanned by some $d_{\tilde{\mathcal{H}}}$ vectors chosen from $\{U\ket{v_k}\}_{k=1}^{2N}$. 
Let $\{U\ket{v_{k_j}}\}_{j = 1}^{d_{\tilde{\mathcal{H}}}}$ denote such $d_{\tilde{\mathcal{H}}}$ vectors. 
Similarly, from Eq.~\eqref{eq:sigma_2_block}, vector space $\tilde{\mathcal{H}}$ is also spanned by some $d_{\tilde{\mathcal{H}}}$ vectors $\{U\ket{u_l}\}_{l=1}^{d_{\tilde{\mathcal{H}}}}$, where $\{\ket{u_l}\}_{l=1}^{2N}$ are linearly independent eigenvectors of $\sigma_2$. 
In particular, since $U\ket{u_1} \in \tilde{\mathcal{H}}$, we may write 
\begin{equation}
\label{eq:condition}
    U\ket{u_1} = \sum_{j=1}^{d_{\tilde{\mathcal{H}}}}c_{k_j}U\ket{v_{k_j}}
\end{equation}
with some $c_{k_j} \in \mathbb{C}$. 
By applying $U^\dagger$ from the left, 
\begin{equation}
    \ket{u_1} = \sum_{j=1}^{d_{\tilde{\mathcal{H}}}}c_{k_j}\ket{v_{k_j}}. 
\end{equation}
Recall 
\begin{equation}
    \ket{v_{k_j}} = 
    \frac{1}{\sqrt{2N}}
    \left(
    \begin{array}{c}
        1   \\
        \omega^{k_j} \\ 
        \omega^{2k_j} \\ 
        \omega^{3k_j} \\ 
        \vdots \\ 
        \omega^{(2N-1)k_j}
    \end{array}
    \right).
\end{equation}
Let $\vec{s}_{k_j}$ denote the vector consisting of the top $N$ entries of $\ket{v_{k_j}}$ 
and 
$\vec{t}_{k_j}$ denote the vector consisting of the bottom $N$ entries of $\ket{v_{k_j}}$; that is, 
\begin{align}
    \vec{s}_{k_j}
    &\coloneqq 
    \frac{1}{\sqrt{2N}}\left(
    \begin{array}{c}
        1   \\
        \omega^{k_j} \\ 
        \omega^{2k_j} \\ 
        \vdots \\ 
        \omega^{(N-1)k_j}
    \end{array}
    \right) \\ 
    \vec{t}_{k_j} 
    &\coloneqq 
    \frac{1}{\sqrt{2N}}\left(
    \begin{array}{c}
        \omega^{Nk_j} \\ 
        \omega^{(N+1)2k_j} \\ 
        \omega^{(N+2)2k_j} \\ 
        \vdots \\ 
        \omega^{(2N-1)k_j}
    \end{array}
    \right), 
\end{align}
and $\ket{v_{j_k}} = \vec{s}_{j_k} \oplus \vec{t}_{j_k}$. 

When $\ket{u_1}$ is given as 
\begin{equation}
    \ket{u_1} = \left(
    \begin{array}{c}
        0 \\
        0 \\
        \vdots \\ 
        0 \\ 
        a_1   \\
        a_2 \\ 
        \vdots  \\ 
        a_{N} 
    \end{array}
    \right), 
\end{equation}
to have Eq.~\eqref{eq:condition}, it is necessary that 
\begin{equation}~\label{eq:condition_necessary}
    \left(
    \begin{array}{c}
        0 \\
        0 \\
        \vdots \\ 
        0  
    \end{array}
    \right) = \sum_{j=1}^{d_{\tilde{\mathcal{H}}}} c_{k_j}\vec{s}_{k_j}. 
\end{equation}
Here, we looked at the top $N$ entries of each vector in Eq.~\eqref{eq:condition}.  
We will prove that $d_{\tilde{\mathcal{H}}} > N$. 
We first show that for any $N$ distinct vectors $\{\ket{v_{k_l}}\}_{l=1}^{N}$ chosen from $\{\ket{v_k}\}_{k=1}^{2N}$, $\{\vec{s}_{k_l}\}_{l=1}^{N}$ are linearly independent. 
For this purpose, it suffices to show that the determinant of matrix 
\begin{equation*}~\label{eq:vandermonde_matrix}
    M \coloneqq \left(
    \begin{array}{ccccc}
        1 & 1 & 1 &\cdots & 1\\
        \omega^{k_1} & \omega^{k_2} & \omega^{k_3} &\cdots&  \omega^{k_N} \\ 
        \omega^{2k_1} & \omega^{2k_2} & \omega^{2k_3} &\cdots&  \omega^{2k_N} \\
        \vdots & \vdots & \vdots & \ddots & \vdots \\ 
        \omega^{(N-1)k_1} & \omega^{(N-1)k_2} & \omega^{(N-1)k_3} &\cdots&  \omega^{(N-1)k_N}
    \end{array}
    \right)
\end{equation*}
is not zero. 
Recalling Vandermonde's determinant, the determinant of $M$ is given as
\begin{equation}
    \mathrm{det}\,M = \prod_{1\leq p<q \leq N} (\omega^{k_q}- \omega^{k_p}). 
\end{equation}
Since $\omega^{k_p} \neq \omega^{k_q}$ for any $1\leq p<q \leq N$, $\mathrm{det}\,M \neq 0$, and thus  $\{\vec{s}_{k_l}\}_{l=1}^{N}$ are linearly independent. 
Now, we show that $d_{\tilde{\mathcal{H}}} > N$.
By way of contradiction, suppose that $d_{\tilde{\mathcal{H}}} \leq N$. 
From the argument above, $\{\vec{s}_{k_j}\}_{j=1}^{d_{\tilde{\mathcal{H}}}}$ in RHS of Eq.~\eqref{eq:condition_necessary} are linearly independent. 
Hence, to have Eq.~\eqref{eq:condition_necessary}, it is necessary that $c_{k_1} = c_{k_2} = \cdots = c_{k_N} = 0$. 
In this case, by Eq.~\eqref{eq:condition}, $a_1 = a_2 = \cdots = a_N = 0$, but this contradicts $(a_{1},a_{2},\ldots,a_{N}) \neq (0,0,\ldots,0)$. 
Therefore, $d_{\tilde{\mathcal{H}}} > N$. 

On the other hand, when $\ket{u_1}$ is given as 
\begin{equation}
    \ket{u_1} = \left(
    \begin{array}{c}
        a_1   \\
        a_2 \\ 
        \vdots  \\ 
        a_{N} \\ 
        0 \\
        0 \\
        \vdots \\ 
        0
    \end{array}
    \right), 
\end{equation}
to have Eq.~\eqref{eq:condition}, it is necessary that 
\begin{equation}~\label{eq:condition_necessary_2}
    \left(
    \begin{array}{c}
        0 \\
        0 \\
        \vdots \\ 
        0  
    \end{array}
    \right) = \sum_{j=1}^{d_{\tilde{\mathcal{H}}}} c_{k_j}\vec{t}_{k_j}. 
\end{equation}
Here, we looked at the bottom $N$ entries of each vector in Eq.~\eqref{eq:condition}.  
With a similar argument, we also have $d_{\tilde{\mathcal{H}}} > N$ in this case. 
Indeed, 
by observing the determinant of matrix 
\begin{equation*}
    \left(
    \begin{array}{ccccc}
        \omega^{Nk_1} & \omega^{Nk_2} & \omega^{Nk_3} &\cdots & \omega^{Nk_N}\\
        \omega^{(N+1)k_1} & \omega^{(N+1)k_2} & \omega^{(N+1)k_3} &\cdots&  \omega^{(N+1)k_N} \\ 
        \omega^{(N+2)k_1} & \omega^{(N+2)k_2} & \omega^{(N+2)k_3} &\cdots&  \omega^{(N+2)k_N} \\
        \vdots & \vdots & \vdots & \ddots & \vdots \\ 
        \omega^{(2N-1)k_1} & \omega^{(2N-1)k_2} & \omega^{(2N-1)k_3} &\cdots&  \omega^{(2N-1)k_N}
    \end{array}
    \right)
\end{equation*}
is given as 
\begin{equation}
   \omega^{N(k_1 + k_2 + \cdots + k_N)} \prod_{1\leq p<q \leq N} (\omega^{k_q}- \omega^{k_p}) \neq 0, 
\end{equation}
for any $N$ distinct vectors $\{\ket{v_{k_l}}\}_{l=1}^{N}$ chosen from $\{\ket{v_k}\}_{k=1}^{2N}$, $\{\vec{t}_{k_l}\}_{l=1}^{N}$ are linearly independent, 
and this leads to $d_{\tilde{\mathcal{H}}} > N$. 

By applying the same argument, we also have $d_{\mathcal{H}'} > N$. 
Hence, $d_{\tilde{\mathcal{H}}} + d_{\mathcal{H}'} > 2N$, which contradicts $d_{\tilde{\mathcal{H}}} + d_{\mathcal{H}'} = 2N$. 
Therefore, $\sigma_1$ and $\sigma_2$ cannot be expressed as Eqs.~\eqref{eq:sigma_1_block} and \eqref{eq:sigma_2_block}; 
that is, $\sigma_1$ and $\sigma_2$ do not have a shared non-trivial block-diagonal structure. 

We proved that $\sigma_1$ and $\sigma_2$ may only have a single-block KI decomposition. 
Now we show that $\sigma_1$ and $\sigma_2$ do not have a single-block redundant part. 
By way of contradiction, 
suppose that we may express 
\begin{align}
    \label{eq:sigma1_redundant}
    \sigma_1 
    &= U (\tilde{\sigma}_1 \otimes \omega) U^\dagger , \\ 
    \label{eq:sigma2_redundant}
    \sigma_2  
    &= U (\tilde{\sigma}_2 \otimes \omega)U^\dagger
\end{align}
using some unitary matrix $U$ and density matrices $\tilde{\sigma}_1$, $\tilde{\sigma}_2$, and $\omega$. 
We assume $\omega$ is $d\times d$ matrix with $d\geq 2$ since otherwise $\omega = 1$ and this will not serve as a redundant part. 
Let $V$ be the unitary matrix corresponding to the permutation of the first and the second systems in Eqs.~\eqref{eq:sigma1_redundant} and \eqref{eq:sigma2_redundant}. 
We may rewrite Eqs.~\eqref{eq:sigma1_redundant} and \eqref{eq:sigma2_redundant} as 
\begin{align}
    \label{eq:sigma1_redundant_rev}
    \sigma_1 
    &= UV (\omega \otimes \tilde{\sigma}_1) V^\dagger U^\dagger , \\ 
    \label{eq:sigma2_redundant_rev}
    \sigma_2  
    &= UV (\omega \otimes \tilde{\sigma}_2)V^\dagger U^\dagger. 
\end{align}
Suppose that $\omega$ is diagonalized as 
\begin{equation}~\label{eq:diag_omega}
    \omega = U_{\omega} \left(
    \begin{array}{cccc}
    \lambda_1 &        &   &   \\
      & \lambda_2 &   &   \\
      &        & \ddots& \\
      &        &   & \lambda_d
    \end{array}
    \right)
    U_{\omega}^\dagger 
\end{equation}
using some $d\times d$ unitary matrix $U_{\omega}$ and the eigenvalues $\{\lambda_j\}_{j=1}^d$ of $\omega$. 
Combining Eqs.~\eqref{eq:sigma1_redundant_rev}, \eqref{eq:sigma2_redundant_rev}, 
and \eqref{eq:diag_omega}, 
\begin{align}
    \label{eq:sigma1_redundant_block}
    \sigma_1 
    &= UV(U_{\omega}\otimes I)\left(\bigoplus_{j=1}^{d} \lambda_j \tilde{\sigma}_1 \right) (U_{\omega}\otimes I)^\dagger V^\dagger U^\dagger \\ 
    \label{eq:sigma2_redundant_block}
    \sigma_2
    &= UV(U_{\omega}\otimes I)\left(\bigoplus_{j=1}^{d} \lambda_j \tilde{\sigma}_2\right)(U_{\omega}\otimes I)^\dagger V^\dagger U^\dagger,  
\end{align}
where $I$ is the identity matrix on the system $\tilde{\sigma}_1$ and $\tilde{\sigma}_2$ reside in. 
Here, $UV(U_{\omega}\otimes I)$ is also unitary. 
Since $d\geq 2$, from Eqs.~\eqref{eq:sigma1_redundant_block} and \eqref{eq:sigma2_redundant_block}, $\sigma_1$ and $\sigma_2$ have a shared non-trivial block-diagonal structure, which is contradiction. 
Therefore, $\sigma_1$ and $\sigma_2$ cannot be expressed as Eqs.~\eqref{eq:sigma1_redundant} and \eqref{eq:sigma2_redundant}, and thus $\sigma_1$ and $\sigma_2$ do not have a redundant part. 
\end{document}